\documentclass[journal]{IEEEtran}

\usepackage{amsmath, bm, authblk, nomencl, lipsum, chngcntr, blindtext}
\usepackage{tikz}
\usepackage[utf8]{inputenc}
\usepackage{pbox}
\usepackage{cite}
\usepackage[english]{babel}
\usepackage{array,colortbl,multirow,multicol,booktabs,ctable}

\usepackage{placeins}
\usepackage{forest}
\usetikzlibrary{arrows,decorations.pathmorphing,backgrounds,fit,positioning,shapes.symbols,chains, shapes,decorations}
\usepackage{csquotes}
\usepackage{comment}
\usepackage{color}
\usepackage{amsfonts}
\usepackage{array}
\usepackage{url}

\usepackage{breakurl}
\usepackage[breaklinks]{hyperref}
\usetikzlibrary{arrows.meta}
\usetikzlibrary{shapes,arrows}
\usetikzlibrary{automata, positioning}
\usepackage{mathtools}
\usepackage[makeroom]{cancel} 

\usetikzlibrary{arrows,decorations.pathmorphing,backgrounds,fit,positioning,shapes.symbols,chains, shapes,decorations}

\title{Privacy-Aware Load Ensemble Control: A Linearly-Solvable MDP Approach}
\author{Ali Hassan, Deepjyoti Deka and Yury Dvorkin}
\date{} 

\usepackage{amsthm}
\theoremstyle{plain}
\newtheorem{theorem}{Theorem}

\newtheorem{proposition}{Proposition}
\newtheorem{corollary}{Corollary}
\theoremstyle{definition}

\usepackage{placeins}
\newcommand{\subparagraph}{}
\usepackage{titlesec}
\titlespacing{\section}{0pt}{3pt}{1pt}
\titlespacing{\subsection}{0pt}{3pt}{1pt}
\usepackage{caption}
\theoremstyle{definition}
\newtheorem{definition}{Definition}[section]
 
\theoremstyle{remark}
\newtheorem*{remark}{Remark}

\begin{document}			
\clearpage
\thispagestyle{empty}
\maketitle

\begin{abstract}
Demand response (DR) programs engage distributed demand-side resources, e.g., controllable residential and commercial loads, in providing ancillary services for electric power systems. Ensembles of these resources can help  reducing system load peaks and meeting operational limits by adjusting their electric power consumption. To equip utilities or load aggregators with adequate decision-support tools for ensemble dispatch, we develop a Markov Decision Process (MDP) approach to optimally control load ensembles in a privacy-preserving manner. To this end, the concept of differential privacy is internalized into the MDP routine to protect transition probabilities  and, thus,  privacy of DR participants. The proposed approach also  provides a trade-off between solution optimality and privacy guarantees, and is analyzed using real-world data from DR events in the New York University microgrid in  New York, NY.
\end{abstract}


\textcolor{black}{
\section*{Nomenclature}
\subsection{Sets and Indices}
\addcontentsline{toc}{subsection}{Sets and Indices}
\begin{IEEEdescription}[\IEEEusemathlabelsep\IEEEsetlabelwidth{ve long}]
\item[$\mathcal{A}$]{Set of discretized  space states}
\item[$\alpha,\beta,\nu$]{States in  set $\mathcal{A}$}
\item[$\mathcal{T}$]{Set of operating time intervals, indexed by $t$}
\item[$\Delta_n$]{Output space of the Dirichlet mechanism}
\item[$\Omega_1,\Omega_2$]{Two disjoint sets that constitute  output space $\Delta_n$}
\end{IEEEdescription}
\subsection{Parameters and Variables}
\addcontentsline{toc}{subsection}{Parameters and Variables}
\begin{IEEEdescription}[\IEEEusemathlabelsep\IEEEsetlabelwidth{ve long}]
\item[$\overline{\mathcal{P}}$]{Non-private default transition probability matrix}
\item[$\overline{\mathcal{P}}^{\alpha\beta}$]{Elements of  non-private default transition probability matrix $\overline{\mathcal{P}}$}
\item[$\overline{\mathcal{\zeta}}^{\beta}$]{Row vector of  non-private default transition probability matrix $\overline{\mathcal{P}}$}
\item[$\overline{\mathcal{\eta}}^{\beta}$]{Vector adjacent to  row vector $\overline{\mathcal{\zeta}}^{\beta}$}
\item[$\Tilde{\overline{\mathcal{P}}}$]{Private default transition probability matrix}
\item[$\Tilde{\overline{\mathcal{P}}}^{\alpha\beta}$]{Elements of  private default transition probability matrix $\Tilde{\overline{\mathcal{P}}}$}
\item[$\gamma$]{Penalty parameter for discomfort cost}
\item[$h$]{Maximum scalar difference between two vectors}
\item[$k$]{Scaling parameter of the Dirichlet mechanism}
\item[$\psi$]{Parameter defining sets $\Omega_1,\Omega_2$ in output space $\Delta_n$}
\item[$\omega,\overline{\omega}$]{Auxiliary parameters of the Dirichlet mechanism}
\item[$\overline{\chi}^{\beta}$]{Output vector of the Dirichlet mechanism}
\item[$\mathcal{P}_t$]{Non-private controlled transition probability matrix}
\item[$\mathcal{P}_t^{\alpha\beta}$]{Elements of  non-private controlled transition probability matrix $\mathcal{P}_t$}
\item[$\Tilde{\mathcal{P}}_t$]{Private controlled transition probability matrix}
\item[$\Tilde{\mathcal{P}}_t^{\alpha\beta}$]{Elements of  private controlled transition probability matrix $\Tilde{\mathcal{P}}_t$}
\item[$\rho_t$]{Non-private vector containing probabilities of the load ensemble being in a certain state such as $\alpha,\beta\in\mathcal{A}$ at time $t$}
\item[$\rho_t^{\alpha\beta}$]{Elements of  non-private vector $\rho_t$}
\item[$\Tilde{\rho}_t$]{Private vector containing probabilities of the load ensemble being in a certain state such as $\alpha,\beta\in\mathcal{A}$ at time $t$}
\item[$\Tilde{\rho}_t^{\alpha\beta}$]{Elements of  private vector $\Tilde{\rho}_t$}
\item[$U_t^{\beta}$]{Utility of the aggregator in state $\beta$ at time $t$}
\item[$\Delta C$]{Cost of privacy}
\end{IEEEdescription}
\subsection{Functions and Operators}
\addcontentsline{toc}{subsection}{Functions and Operators}
\begin{IEEEdescription}[\IEEEusemathlabelsep\IEEEsetlabelwidth{ve long}]
\item[$\varphi_t^{\beta}$]{Non-private value function of state $\beta$ at time $t$}
\item[$z_t^{\beta}$]{Non-private desirability function of state $\beta$ at time $t$}
\item[$\Tilde{\varphi}_t^{\beta}$]{Private value function of state $\beta$ at time $t$}
\item[$\Tilde{z}_t^{\beta}$]{Private desirability function of state $\beta$ at time $t$}
\item[$\text{B}(.)$]{Multivariate beta function}
\item[$\Gamma(.)$]{Gamma function}
\item[$\psi(.)$]{Digamma function}
\item[$KL\text{[.]}$]{Kullback–Leibler divergence}
\item[$\mathbb{P}\text{[.]}$]{Probability of a random variable}
\item[$\mathbb{E}\text{[.]}$]{Expectation of a random variable}
\item[$\text{Var}(.)$]{Variance of a random variable}
\end{IEEEdescription}
}

\section{Introduction}

Electric power utilities deploy Demand Response (DR) programs to reduce dependency on conventional generation assets for meeting peak electricity demand or avoiding network congestion. DR programs provide an alternative to using these resources by encouraging electricity customers to change their power consumption in exchange for an incentive \cite{DR_overview}. This incentive can come in a form of direct monetary compensations (e.g. payment or electricity bill rebate/credit) or else (e.g. voucher, gift card, etc) \cite{DR_incentive_giftcards,DR_incentive_giftcards2}. In total, the U.S. Energy Information Administration (EIA) reports that DR programs yielded energy savings of 1,462,735 MWh in 2019, an increase of 9.4\% from 2016 \cite{EIA}.

To model residential DR programs, which must deal with a large number of small-scale appliances with various nameplate parameters and characteristics, the Markov Decision Process (MDP) framework has been adopted and then field-validated as an effective and scalable decision-making framework for controlling large ensembles of such appliances and buildings \cite{Callaway_TCL,MDP_Meyn,MDP_Meyn2,MDP_Meyn3,Chertkov_MDP_chap,MDP_Deka,MDP_Hassan,MDP_Pop,MDP_Emiliano,MDP_Turitsyn,MDP_Hassan_Uncertainty,MDP_Learning_Ali}. The advantage of this framework is in connecting (i) modeling convenience of the MDP, which makes it easier to obtain high-fidelity and data-driven representations of load and building ensembles, and (ii) solution simplicity of dynamic programming. However, while the MDP is useful to operators of DR programs \cite{Callaway_TCL,MDP_Meyn,MDP_Meyn2,MDP_Meyn3,Chertkov_MDP_chap,MDP_Deka,MDP_Hassan,MDP_Pop,MDP_Emiliano,MDP_Turitsyn,MDP_Hassan_Uncertainty,MDP_Learning_Ali}, it has relatively large data requirements to properly construct a Markov Process (MP) underling the MDP. This, in turn, raises privacy and security concerns that the data used for constructing a MP of the load ensemble can be breached and sensitive features of individual DR participants can be revealed \cite{Privacy_issues,Privacy_NIST}. 

Privacy issues are not unique to DR programs. For example, in 2007, Netflix released a data set of their user ratings as part of a competition to see if anyone can outperform their collaborative filtering algorithm. Although all personal identifying information was removed for privacy protection, two researchers from the University of Austin were still able to breach privacy and recovered 99\% of the removed personal information using reference data from IMDb \cite{Netflix}. This exposed the ability of compromising user privacy by reconstructing useful information from other, potentially unprotected, data sets. Similarly, large volumes of data generated and gathered for DR programs can be used to compromise the privacy of  DR participants and reveal their sensitive information. For example, detailed electricity consumption data can reveal lifestyle profiles and choices, including but not limited to a number of people living in a household, absence or presence of household members, working hours, meal times, types of appliances and even religious affiliations in some cases \cite{Privacy_leaks}. National Institute of Standards and Technology (NIST) considers four different dimensions or aspects of privacy as -- personal information, personal privacy, behavioral privacy, and personal communication -- that can be infringed upon by using the electricity consumption and DR program data \cite{Privacy_NIST}.

\textcolor{black}{Accordingly, motivated by the data breaches as in \cite{Privacy_issues,Privacy_NIST,Privacy_leaks}, DR programs should be immunized against risks of exposing their customers. Previously, privacy concerns in DR programs have been studied in different contexts. 
The notable techniques employed to protect the privacy involve data aggregation \cite{Data_Agg_1,Data_Agg_2,Data_Agg_3,Data_Agg_4}, data anonymization \cite{Data_Ann_1,Data_Ann_2,Data_Ann_3} and data perturbation \cite{Data_pert_1,Data_pert_2,Data_pert_3}. Data aggregation and anonymization schemes can provide a false sense of privacy. For example, even though they can ensure the security of data during communication, privacy could still be compromised by specific filters such as Collaborative Filtering \cite{Privacy_filter} or reconstruction theorems or attacks \cite{reconstruction}, and the aggregated data can be manipulated to release sensitive information about individuals (see the Netflix data breach example mentioned above \cite{Netflix}. In \cite{Data_Jawurek}, Jawurek et al. show that even after pseudonymization, which is a  data management procedure promoted by the European Union's General Data Protection Regulation (GDPR), data items can still be traced to their respective individuals sources. Data perturbation means adding random noise to the data to protect the original data, but this procedure per se does not always guarantee privacy as original data can be recovered from the perturbed data \cite{Data_pert_attack1} (e.g. by means of  de-noising \cite{Data_pert_attack2}). This motivates us to explore such techniques that can guarantee privacy, while ensuring the usability of the aggregated data.} Differential privacy (DP) is a mathematically rigorous framework that allows for sharing aggregated data (e.g. parameters of a given load ensemble), while preserving features of individuals (e.g. loads in the ensemble). In other words, DP makes it possible to quantify to what extent privacy of an individual in the aggregated data is preserved, while sharing useful aggregated information about the data set \cite{dwork2008differential}. This property, as explained in \cite{Dwork_book}, ensures that DP is capable of \enquote{learning nothing about an individual, while learning useful information about a population}. This property makes DP suitable for internalizing privacy restrictions into the MDP framework applied to DR programs, which deals with the aggregated behavior of DR participants. Additionally, DP performance hardly changes when a single individual joins or leaves a given ensemble, i.e. its output will be roughly the same, regardless if an individual is part of the data set or not. Notably, despite being a fairly recent invention, DP has already been widely used in  privacy-critical applications, e.g. by US Census Bureau \cite{DP_Census}, Apple \cite{DP_Apple}, LinkedIn \cite{DP_LinkedIn}, Facebook \cite{DP_Facebook}, Microsoft \cite{DP_Microsoft} and Google \cite{DP_Google}.

A DP analysis of a method or an algorithm is usually called a mechanism. Any DP mechanism is randomized to ensure that anyone observing its output will essentially make the same inference about any individual’s data whether or not that individual’s data is included in its input. As such, a given mechanism is usually randomized by adding some noise to the input or output of the algorithm depending upon the given framework or application \cite{chaudhuri2011differentially}. Some common DP mechanisms include Laplace and Gaussian mechanisms, where added noises are drawn from Laplace and Gaussian distributions. While these mechanisms are widely employed in practice, their application in the MDP is challenging as adding a random noise may not comply with the integrality requirement on the MDP, i.e. the cumulative probability of moving from any given state to all future possible states is equal to one. Recently, this challenge has been overcome by means of the Dirichlet mechanism \cite{Dirichlet_mechanism}, which was subsequently introduced into the MDP \cite{Dirichlet_mechanism_MDP}. Since the Dirichlet distribution is simplex-constrained, it allows for maintaining the integrality requirement for the MDP. However, the result in \cite{Dirichlet_mechanism_MDP} deals with conventional MDPs, i.e. it searches over all actions for each future state, which is computationally demanding due to the exponential growth of future states. Furthermore, MDPs often do not allow for analytical solutions and require using numerical methods. To overcome this challenge, a conventional MDP can be reduced to a linearly-solvable MDP (LS-MDP), under suitable action costs. \textcolor{black}{Unlike in a conventional MDP, the optimal policy derived from the LS-MDP is not a mapping of states to action variables, but rather a mapping of a current state into a next-state distribution. Additionally, LS-MDPs are reduced to linear eigenvalue problems and  can be solved using analytical optimal policy, thus reducing computational requirements, while providing accurate approximations to conventional MDPs \cite{LSMDP_todorov}.} 

\textcolor{black}{The current literature on DP implementation for DR programs adds noise to the input or aggregated data \cite{Data_diff_privacy1,Data_diff_privacy2,Data_diff_privacy3}. Such methods don't consider the optimization problem, limiting their application since their solutions are not certified feasible. In contrast, we internalize the noise into our optimization problem and utilize stochastic programming \cite{Valadmir2} for privacy-aware DR solutions.}
This paper introduces DP into a LS-MDP framework, where DP is employed to privatize default transition probabilities, which are obtained from the MP derived from electricity consumption data, by means of adding a noise from a specified Dirichlet distribution to protect the privacy of DR participants.
\textcolor{black}{The main contributions are as follows:
\begin{enumerate}
\item We leverage stochastic programming to derive privatized optimal policies. These optimal policies are the generalized expressions of the stochastic policies derived in \cite{MDP_Hassan_Uncertainty} to internalize uncertainty, using first and second statistical moments, but extended to privatized versions with privacy guarantees on default transition probabilities. 
\item While the use of stochastic programming enables the aggregators to protect the privacy of DR participants in their optimization problems, they may need to send their data to external entities (e.g. web/cloud computing or report data to authorities for audits \cite{LL84_LL87}). This motivates us for exploring the average value method for DP, our second contribution. This method solves the LS-MDP problem for a set of randomly generated samples of default transition probabilities, and provides randomized optimal policies for these samples. This also allows the aggregator to outsource the optimization step to a third entity without compromising the privacy of DR participants. Unlike in the stochastic methods described above, the average value method does not require specifying first and second statistical moments \textcolor{black}{as exactly known for the randomized policies.} The average of all the optimal randomized policies is computed to compare the outcome of the average value method to the stochastic approach.
\item The cost of privacy for the DP LS-MDP is also derived for both the stochastic and average value methods, and the loss of optimality, also known as the cost of privacy, that occurs as the result of providing privacy to DR participants, is analyzed.
\end{enumerate}}
\textcolor{black}{The usefulness of the proposed framework is demonstrated on the real-life data based on New York University (NYU) campus building, which participates in Consolidated Edison operated DR events.}

\textcolor{black}{The rest of this paper is organized as follows. Section II presents the fundamentals of LS-MDP and derives a non-private control policy. Section III introduces the basic concepts of differential privacy, describes the Dirchlet mechanism and related properties. Sections IV and V describe the proposed stochastic and average values approaches respectively and derive privacy-aware control policies and their cost of privacy expressions. Section VI presents the case study to show the effectiveness of the proposed models. Section VII concludes the paper.}

\section{A LS-MDP Formulation} \label{sec:ls-mdp_formulation}
The load ensemble can be represented by a discrete-time, discrete-space MDP similar to our prior work in \cite{Chertkov_MDP_chap,MDP_Deka,MDP_Hassan,MDP_Pop,MDP_Hassan_Uncertainty,MDP_Learning_Ali}. We consider a LS-MDP with state space $\mathcal{A}$, time space $\mathcal{T}$, default transition probability matrix $\overline{\mathcal{P}}\in \mathbb{R}^{n\times n}$, $n=|\mathcal{A}|$, controlled transition probability matrix $\mathcal{P}_{t}\in \mathbb{R}^{n\times n}$, $n=|\mathcal{A}|$ and utility of the aggregator $U_{t}^{\beta}$. State space $\mathcal{A}=\{\alpha,\beta, \cdots\}$ is obtained by discretizing controllable parameters of the loads in the ensemble (e.g., power consumption or temperature levels) and time space $\mathcal{T}=\{t,t+1, \cdots \}$ represents discrete time intervals (e.g., hours or minutes). Fig.~\ref{fig:MP_representation} illustrates a MP with eight states with possible transitions from state 1 to other states. Default transition probability matrix $\overline{\mathcal{P}}\in \mathbb{R}^{n\times n}$ represents the evolution of the ensemble in the state space without controlled decisions of the aggregator, i.e. it represents a MP that corresponds to the endogenous dynamics of the load ensemble. On the other hand, controlled transition probability matrix $\mathcal{P}_{t}\in \mathbb{R}^{n\times n}$ represents the evolution of the ensemble under the control actions of the aggregator. Notably, matrix $\mathcal{P}_{t}\in \mathbb{R}^{n\times n}$ is time-variable to enable dynamic control of the ensemble. More specifically, $\overline{\mathcal{P}}^{\alpha \beta}$ and ${\mathcal{P}}^{\alpha \beta}_t$ are elements of matrix $\overline{\mathcal{P}}$ and $\mathcal{P}_t$ associated with a transition from state $\beta$ to state $\alpha\neq\beta\in \mathcal A$. 

While matrix elements $\overline{\mathcal{P}}^{\alpha \beta}$ can be retrieved from processing historical observations (measurements) of the load ensemble \cite{Callaway_TCL}, matrix elements $\mathcal{P}_{t}^{\alpha\beta}$ are the outcomes of the decision-making process of the aggregator, which aims to maximize its utility function given by $U_{t}^{\beta}$. This decision-making process can be cast as: 
\begin{figure}[!t]
\centering 
\includegraphics[width=0.95\columnwidth]{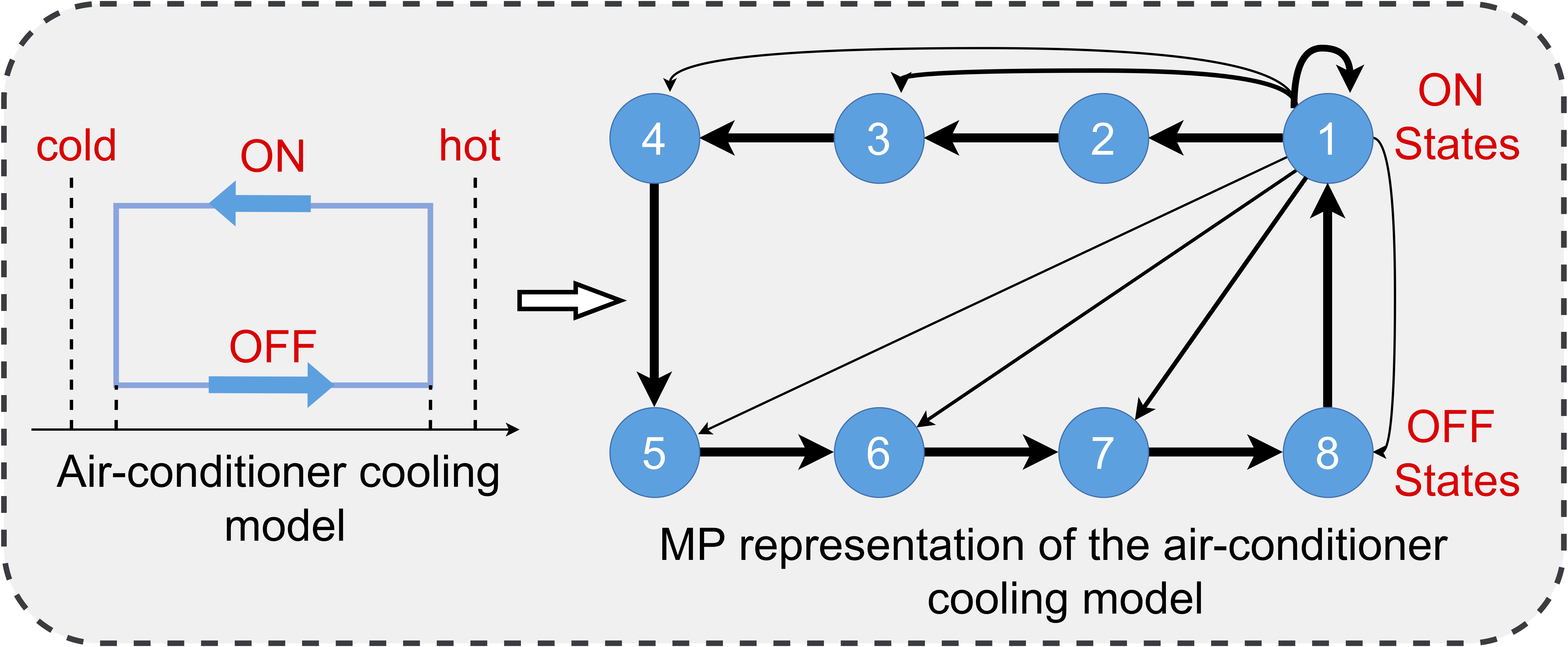}
\caption{A MP representation of a load ensemble with 8 states displaying all possible transitions from state 1.}
\label{fig:MP_representation}
\end{figure}
\begin{subequations}
\begin{align}
&\underset{\substack{\rho_t,\mathcal{P}_t}}{\text{min}} O^{A}:=  \mathbb{E}_{\rho}\!\!
\sum_{t \in \mathcal{T}-1} \sum_{\alpha \in \mathcal{A}} \Big(-U_{t+1}^{\alpha} + \sum_{\beta \in \mathcal{A}} \gamma \log\! \frac{\mathcal{P}_{t}^{\alpha\beta}}{\overline{\mathcal{P}}^{\alpha\beta}}\Big) \label{MDP:obj}
\\
&\text{s.t.} \ \ \ \rho_{t+1}^{\alpha} = \sum_{\beta \in \mathcal{A}} \mathcal{P}_{t}^{\alpha\beta} \rho_{t}^{\beta}, \quad \forall \alpha \in \mathcal{A}, t \in \mathcal{T} 
\label{MDP_evol} \\
&\ \ \ \ \ \sum_{\alpha \in \mathcal{A}} \mathcal{P}_{t}^{\alpha \beta} = 1,\quad \forall \beta \in \mathcal{A}, t \in \mathcal{T}-1, \label{mdp_integrality} 
\end{align}
\label{base_mdp}
\end{subequations}
where $\rho_t \in \mathbb{R}^{n}$, $n=|\mathcal{A}|$, is a vector with entries $\rho_{t+1}^{\alpha} \geq 0$ and $\rho_{t}^{\beta} \geq 0$ representing the probabilities that the load ensemble is operated in states $\alpha,\beta\in\mathcal{A}$ at times $t+1$ and $t$, respectively, and are related via controlled transition probabilities $\mathcal{P}_{t}^{\alpha\beta}$. Eq.~\eqref{MDP:obj}, where operator $\mathbb{E}_{\rho}$ denotes the expectation over $\rho$, maximizes the utility of the DR aggregator and minimizes the discomfort cost of DR participants. The discomfort cost is modeled as the  Kullback-Leibler (KL) divergence \cite{KL_book}, i.e. the logarithmic difference between the default transition probabilities and  controlled transition probabilities weighed by cost penalty $\gamma$, which can be interpreted as a lost utility due to the need to deviate from the default behavior of electric loads in the ensemble. For example, this interpretation can be viewed as a change in temperature settings of air-conditioning units from user-defined comfort levels to the levels optimized by the DR aggregator. Furthermore, if $\mathcal{\overline{P}}^{\alpha \beta}=0$, i.e. a transition from state $\beta$ to $\alpha$ has not been observed in the past, the model in \eqref{base_mdp} restricts $\mathcal{P}_{t}^{\alpha\beta}=0$ and excludes such transitions from the optimization \footnote{\textcolor{black}{Including historically unobserved but still probable transitions will require either re-constructing the MP to reflect these default transition probabilities or learning them via reinforcement learning techniques, see  \cite{MDP_Learning_Ali}.}}. Eq.~\eqref{MDP_evol} describes the temporal evolution of the load ensemble (e.g., changes in temperature settings of air-conditioning units) from time $t$ to $t + 1$ over time horizon $\mathcal{T}$. Eq.~\eqref{mdp_integrality} imposes the integrality constraint on the transition decisions optimized by the DR aggregator such that their total probability is equal to one. As shown in \cite{MDP_Hassan_Uncertainty}, Eq.~\eqref{base_mdp} can be optimized as follows:

\begin{theorem} \label{theorem_0} \normalfont
Let \eqref{base_mdp} model a load ensemble as a LS-MDP. Then, the optimal control policy $\forall \beta,\alpha\in\mathcal{A}, \forall t\in\mathcal{T},$ is computed as:
\begin{align}
\begin{split}
&\mathcal{P}_{t}^{\alpha \beta} = \frac{\overline{\mathcal{P}}^{\alpha\beta}z^{\alpha}_{t+1}}{\sum\limits_{\alpha\in\mathcal{A}}\overline{\mathcal{P}}^{\alpha\beta}z^{\alpha}_{t+1}}, \label{optimal_policy}
\end{split}
\end{align}
where $\forall \beta\in\mathcal{A}, \forall t\in\mathcal{T},$ and 
\begin{align}
&z^{\beta}_{t} = \text{exp}\Big(\frac{U_{t}^{\beta}}{\gamma}\Big) \sum_{\alpha\in\mathcal{A}}\overline{\mathcal{P}}^{\alpha\beta}z^{\alpha}_{t+1}, \label{bellmen_blackuced_det}
\end{align}
\end{theorem}
\begin{proof} 
See proof in Appendix \ref{appendix_proof_theorems}.
\end{proof}
\textcolor{black}{Here, $z_{t}^{\beta}$ is a desirability function representing value function $\varphi^{\beta}_{t}$ of state $\beta$ at time $t$, i.e., $z_{t}^{\beta} = \text{exp}(\frac{-\varphi_{t}^{\beta}}{\gamma})$. Similarly, $z_{t+1}^{\alpha}$ is a desirability function for state $\alpha$ at time $t+1$.}

\section{Differentially Private LS-MDP} \label{sec:dp_ls_mdp}
The LS-MDP developed in Section \ref{sec:ls-mdp_formulation} to model an ensemble of controllable loads is prone to privacy breaches and can reveal sensitive data of DR participants. In this section, we incorporate DP into the LS-MDP formulation in \eqref{base_mdp} to protect privacy of DR participants and derive differentially private optimal policy as summarized in Fig.~\ref{fig:DP_LS-MDP}. We explain the use of DP for LS-MDPs by starting with the following definitions:

\begin{definition}[Adjacent vectors] \label{def:adjacent_vectors}
Let $\overline{\mathcal{\zeta}}^{\beta}\in \mathbb{R}^{n}, n=|\mathcal{A}|,$ be a row vector of matrix $\overline{\mathcal{P}}$, i.e. $\overline{\mathcal{\zeta}}^{\beta} =\{\overline {\mathcal P}^{\alpha\beta}, \forall \alpha \in \mathcal{A} \}$, and let vector $\overline{\mathcal{\eta}}^{\beta}\in \mathbb{R}^{n}, n=|\mathcal{A}|,$ have the same elements as vector $\overline{\mathcal{\zeta}}^{\beta}$ except for two randomly chosen elements. Then vectors $\overline{\mathcal{\zeta}}^{\beta}$ and $\overline{\mathcal{\eta}}^{\beta}$ are adjacent if:
\begin{align}
& ||\overline{\zeta}^{{\beta}} - \overline{\mathcal{\eta}}^{\beta}||_{1} \leq h,
\end{align}
\textcolor{black}{where $||\cdot||_{1}$ denotes the first norm of a vector and $0<h<1$ represents the maximum scalar difference between the two vectors. Using $h=0$ will result in two identical vectors i.e., $\overline{\zeta}^{{\beta}} = \overline{\mathcal{\eta}}^{\beta}$ , and $h=1$ will violate the integrality constraint on the row vector $\overline{\mathcal{\eta}}^{\beta}$.} 
We denote  two adjacent vectors  using symbol  $\simeq$, i.e. $\overline{\mathcal{\zeta}}^{\beta}\simeq\overline{\mathcal{\eta}}^{\beta}$. 
\end{definition} 
Consistently with \cite{Dirichlet_mechanism,Dirichlet_mechanism_MDP}, Definition \ref{def:adjacent_vectors} ensures that vectors $\overline{\mathcal{\zeta}}^{\beta}$ and $\overline{\mathcal{\eta}}^{\beta}$ are different by no more than value $h$, which is user defined. This definition departs from the traditional notion of differential privacy \cite{Dwork_book}, where two data sets are different in only one entry. This modification is needed to ensure that if one element in vector $\overline{\mathcal{\eta}}^{\beta}$ is modified to ensure DP, the integrality constraint  in \eqref{mdp_integrality} holds, i.e. the probability of the ensemble being in a given state is equal to one. In other words, we modify two elements in vector $\overline{\mathcal{\eta}}^{\beta}$ because in the LS-MDP context, the vector elements are interpreted as probabilities and, thus, must add up to one. 

\begin{definition}[Differential Privacy] \label{definition_dp}
Random mechanism $\Tilde{\mathcal{M}}$ is $\epsilon$-differentially private, if for all adjacent input vectors $\overline{\mathcal{\zeta}}^{\beta}\simeq\overline{\mathcal{\eta}}^{\beta}$ and  output vectors ${\overline{\chi}}^{\beta}$, the following holds:
\begin{align}
& \mathbb{P}[\Tilde{\mathcal{M}}(\overline{\zeta}^{\beta})={\overline{\chi}}^{\beta}] \leq e^{\epsilon} \mathbb{P}[\Tilde{\mathcal{M}}(\overline{\eta}^{{\beta}}\!)={\overline{\chi}}^{\beta}],
\end{align}
where $\mathbb{P}$ denotes the probability of an event and $\epsilon \geq 0$ represents the privacy loss quantified by DP, where low values of $\epsilon$ provide more privacy at the expense of accuracy.
\end{definition}
\textcolor{black}{As applied to the default transition probabilities of the LS-MDP optimization in \eqref{base_mdp}, Definition~\ref{definition_dp} treats this process as mechanism $\Tilde{\mathcal{M}}$ and its input $\overline{P}^{\alpha \beta}$ and output $\Tilde{\overline{\mathcal{P}}}^{\alpha\beta}$ as components of vectors $\overline{\mathcal{\zeta}}^{\beta}$ and ${\overline{\chi}}^{\beta}$, respectively. Thus, this DP definition guarantees that a randomized mechanism behaves similarly on  adjacent inputs, thus protecting the individual's privacy in the database by not providing enough information on the output of the randomized mechanism that can relate it (or trace back) to a particular input database.}

\begin{definition}[Probabilistic Differential Privacy]
Random mechanism $\Tilde{\mathcal{M}}$ satisfies $(\epsilon,\delta)$-probabilistic differential privacy, if for all vectors $\overline{\zeta}^{\beta}$ output space $\Delta_{n}$ of $\Tilde{\mathcal{M}}$ can be divided into sets $\Omega_1$, $\Omega_2$ such that:
\begin{align}
& \mathbb{P}[\Tilde{\mathcal{M}}(\overline{\zeta}^{\beta}) \in \Omega_{2}] \leq \delta,
\end{align}
and for all $\overline{\zeta}^{{\beta}}\! \simeq \overline{\eta}^{\beta}$ and for all $\overline{\chi}^{\beta} \in \Omega_{1}$:
\begin{align}
& \log \Bigg( \frac{\mathbb{P}[\Tilde{\mathcal{M}}(\overline{\zeta}^{\beta})=\overline{\chi}^{\beta}]}{\mathbb{P}[\Tilde{\mathcal{M}}(\textcolor{black}{\overline{\eta}^{{\beta}}}\!)=\overline{\chi}^{\beta}]} \Bigg) \leq \epsilon.
\end{align}
where $\epsilon \geq 0$, and $0\leq\delta \leq 1$ represents the maximum probability of breaching $\epsilon$-differential privacy.
\end{definition}
This definition of $(\epsilon,\delta)$-probabilistic differential privacy guarantees that $\epsilon$-differential privacy is achieved with at least $(1-\delta)$ probability \cite{Probablisitc_DP}. \textcolor{black}{Both parameters $\epsilon$ and $\delta$ are user-defined and can be selected  to satisfy privacy preferences of a particular decision-maker. A smaller value of $\epsilon$ reflects stronger privacy for a relatively less accurate response, and $\delta$ is usually chosen very small ($\leq 0.05$) to ensure high satisfaction of $\epsilon$-differential privacy.}

\begin{definition}[Dirichlet Mechanism] \label{definition~DM}
Let $\Tilde{\mathcal{M}}$ be a Dirichlet mechanism with input vector $\overline{\zeta}^{\beta}$ and output vector $\overline{\chi}^{\beta}$ according to the Dirichlet distribution, i.e.: 
\begin{align}
\begin{split}
& \mathbb{P}[\Tilde{\mathcal{M}}(k\overline{\zeta}^{\beta})\! =\! {\overline{\chi}}^{\beta}] = \\& \qquad \frac{1}{\text{B}(k\overline{\zeta}^{\beta})} \! \prod_{i=1}^{n-1} ({\overline{\chi}}^{\beta}\!)^{k\overline{\zeta}_{i}^{\beta} - 1} \!\Big(\! 1 -\!\! \sum_{i}^{n-1} {\overline{\chi}}_{i}^{\beta} \Big)^{k\overline{\zeta}_{n}^{\beta} - 1},
\end{split}
\end{align}
where $k \geq 0$ is a user-defined parameter that can be adjusted to provide a trade-off between privacy and the accuracy level of the Dirichlet mechanism, and $\text{B}(k\overline{\zeta}^{\beta})$ is the normalizing constant, which is a multivariate beta function,  expressed as:
\begin{align}
& \text{B}(k\overline{\zeta}^{\beta}) = \frac{\prod_{i=1}^{n}\Gamma(k\overline{\zeta}_{i}^{\beta})}{\Gamma(k\sum_{i=1}^{n}\overline{\zeta}_{i}^{\beta})}
\end{align}
\end{definition}
\noindent where $\Gamma(.)$ is the gamma function \cite{Digamma_func}.
The Dirichlet mechanism in Definition~\ref{definition~DM} is a randomized mechanism that takes a vector scaled by parameter $k$ as an input and provides an output vector according to the Dirichlet distribution.
\textcolor{black}{The output vector is more concentrated around the input vector for greater values of $k$, that is the value of $k$ controls the amount of noise added by the Dirichlet mechanism. When $k$ is small, the output vector tends to depart from the input vector, indicating better privacy, but this comes at the expense of the solution accuracy.}
\textcolor{black}{So, physically the effect of differential privacy on the load ensemble represented by a LS-MDP is its perturbation via Dirichlet noise without compromising its properties. To ensure privacy guarantee for this mechanism, we state:}

\begin{theorem}[Privacy Guarantee \cite{Dirichlet_mechanism}] \label{theorem_dm_privacy} \normalfont
The Dirichlet mechanism in Definition~\ref{definition~DM} satisfies $(\epsilon,\delta)$-probabilistic differential privacy with:
\begin{align}
\begin{split}
& \hspace{-10mm}\epsilon = \log \Bigg( \frac{\text{B}(k\omega,k(1-\overline{\omega}-\omega))}{\text{B}(k(\omega+\frac{h}{2}),k(1-\overline{\omega}-\omega-\frac{h}{2}))} \Bigg) \\& \qquad \qquad \qquad
+ \frac{kh}{2}\log \bigg( \frac{1-(|W|-1)\psi}{\psi} \bigg)
\end{split}\\
&\hspace{-10mm}\text{and} \hspace{12mm} \delta = 1 - \underset{\overline{\zeta}^{\beta}}{\min}\ \mathbb{P}[\Tilde{\mathcal{M}}(\overline{\zeta}^{\beta})\in\Omega_1] 
\end{align}
where $\omega,\overline{\omega},\psi\in(0,1)$, $h\in(0,1]$ and $W \subseteq [n-1]$. Note that output space $\Delta_{n}$ of  Dirichlet mechanism $\Tilde{\mathcal{M}}$ is split in two disjoint sets $\Omega_1 := \{ \overline{\zeta}^{\beta}\in\Delta_{n} | \overline{\mathcal{P}}^{\alpha\beta} \leq \psi \ , \forall \beta \in \mathcal{A} \}$ and $\Omega_2 := \{\Delta_{n}\setminus\Omega_1\}$.
\end{theorem}
\begin{proof}
Detailed proof is presented in \cite{Dirichlet_mechanism}. 
\end{proof}

\begin{remark}
Alternately, $\delta$ in Theorem \ref{theorem_dm_privacy} can be selected as the maximum probability of privacy failure to divide the output space $\Delta_{n}$ of $\Tilde{\mathcal{M}}$ in disjoint sets $\Omega_1$ and $\Omega_2$ based on the value obtained for $\psi$, and $\epsilon$ is calculated accordingly. Interested readers are referred to \cite{Dirichlet_mechanism} for more details.
\end{remark}
Owing to Theorem \ref{theorem_dm_privacy}, the Dirichlet mechanism in Definition~\ref{definition~DM} makes it possible to obtain private default transition probabilities $({\Tilde{\overline{\mathcal{P}}}}^{\alpha\beta})$ to be used in \eqref{base_mdp}, which satisfy $(\epsilon,\delta)$-probabilistic differential privacy and lead to a private optimal policy using the property of post processing. This is possible since differential privacy is immune to post-processing \cite{Dwork_book}: 

\begin{definition}[Post-Processing] \label{post_processing}
Let $\Tilde{\mathcal{M}}$ be a random, $(\epsilon,\delta)$-differential private mechanism. Let $f$ be an arbitrary mapping. Then $f \circ \Tilde{\mathcal{M}}$ is $(\epsilon,\delta)$-differential private.
\end{definition}
As a result of Definition~\ref{post_processing}, if the LS-MDP optimization in \eqref{base_mdp} is fed with default transition probabilities ${\overline{\mathcal{P}}}^{\alpha\beta}$ and applies the Dirichlet mechanism as in Definition~\ref{definition~DM} to the default transition probabilities, then Theorem~\ref{theorem_dm_privacy} and Definition~\ref{post_processing} ensure that the LS-MDP optimization in \eqref{base_mdp} outputs private controlled transition probabilities, which we represent by $\Tilde{\mathcal{P}}_{t}^{\alpha\beta}$ as shown in Fig.~\ref{fig:DP_LS-MDP}. Sections \ref{Section_Stochastic_Approach} and \ref{Section_Avg_Value_Approach} derive differentially private optimal policies using the application of the Dirichlet mechanism.


\begin{figure}[!t]
\centering 
\includegraphics[width=\columnwidth,trim={0 0 100cm 0.1cm}]{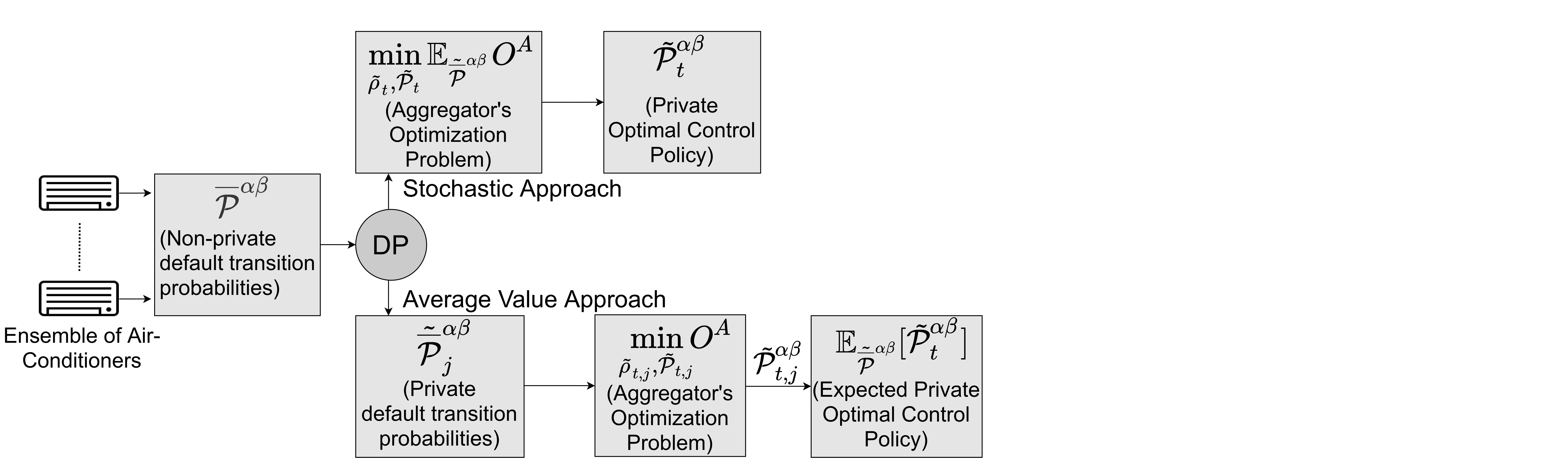}
\caption{LS-MDP with differential privacy}
\label{fig:DP_LS-MDP}
\end{figure}

\section{Stochastic Approach} \label{Section_Stochastic_Approach}
We leverage stochastic optimization to solve the LS-MDP when DP is employed to privatize default transition probabilities. Accordingly, the objective function in Eq.~\eqref{MDP:obj} is re-written for the DP LS-MDP as:
\begin{align}
& \underset{\substack{\Tilde{\rho},\Tilde{\mathcal{P}}}}{\text{min}} \ \mathbb{E}_{{\Tilde{\overline{\mathcal{P}}}^{\alpha\beta}}} \mathbb{E}_{\Tilde{\rho}}
\!\!\sum_{t \in \mathcal{T}}\! \! \sum_{\alpha \in \mathcal{A}}\!\! \bigg(\!\!-U_{t+1}^{\alpha}\! +\!\! \sum_{\beta \in \mathcal{A}}\! \gamma \log\! \frac{\Tilde{\mathcal{P}}_{t}^{\alpha\beta}}{\Tilde{\overline{\mathcal{P}}}^{\alpha\beta}}\bigg) = \nonumber \\
&\!\!\!\underset{\substack{\Tilde{\rho},\Tilde{\mathcal{P}}}}{\text{min}} \ \mathbb{E}_{\Tilde{\rho}}\!\! \sum_{t \in \mathcal{T}} \! \! \sum_{\alpha \in \mathcal{A}}\!\!\! \Big\{\!\!-\!U_{t+1}^{\alpha}\! +\!\gamma\!\! \sum_{\beta \in \mathcal{A}}\!\! \big( \log\! \Tilde{\mathcal{P}}_{t}^{\alpha\beta}\!\!-\!\mathbb{E}_{\Tilde{\overline{\mathcal{P}}}^{\alpha\beta}}\! [\log{\Tilde{\overline{\mathcal{P}}}^{\alpha\beta}}] \big)\!\! \Big\} \label{MDP:obj_diff} 
\end{align}
Given \eqref{MDP:obj_diff}, the optimization in \eqref{base_mdp} is recast as:
\begin{subequations}
\begin{align}
 & \text{Eq. }\eqref{MDP:obj_diff} \\
 &\text{s.t.} \ \ \ \text{Eq. } \eqref{MDP_evol}-\eqref{mdp_integrality},
\end{align} \label{diff_priv_mdp}
\end{subequations}
where $\Tilde{\overline{\mathcal{P}}}^{\alpha\beta}$ represents privatized default transition probabilities and the decision-maker employing \eqref{diff_priv_mdp} can select a certain policy on how to randomize the original default transition probabilities. For the reasons stated in Section~\ref{sec:dp_ls_mdp}, we use the Dirichlet mechanism for this randomization as given in Definition~\ref{definition~DM}.
To solve the DP LS-MDP formulation in \eqref{diff_priv_mdp}, we prove: 

\begin{theorem} \label{theorem_a} \normalfont
Let \eqref{diff_priv_mdp} model a load ensemble as a differential private LS-MDP. Then, the differentially private optimal control policy $\forall \beta,\alpha\in\mathcal{A}, \forall t\in\mathcal{T},$ is computed as:
\begin{align}
\begin{split}
&\Tilde{\mathcal{P}}_{t}^{\alpha \beta} \! = \frac{\exp(\mathbb{E}_{\Tilde{\overline{\mathcal{P}}}^{\alpha\beta}} [\log {\Tilde{\overline{\mathcal{P}}}^{\alpha\beta}}\!]) \Tilde{z}^{\alpha}_{t+1}}{\sum\limits_{\alpha\in\mathcal{A}}\exp(\mathbb{E}_{\Tilde{\overline{\mathcal{P}}}^{\alpha\beta}} [\log {\Tilde{\overline{\mathcal{P}}}^{\alpha\beta}}\!])\Tilde{z}^{\alpha}_{t+1}}, \label{optimal_policy_dp}
\end{split}
\end{align}
where $\forall \beta\in\mathcal{A}, \forall t\in\mathcal{T}$:
\begin{align}
&\Tilde{z}^{\beta}_{t} = \text{exp}\Big(\frac{U_{t}^{\beta}}{\gamma}\Big)\!\! \sum_{\alpha}\!\exp(\mathbb{E}_{\Tilde{\overline{\mathcal{P}}}^{\alpha\beta}} [\log {\Tilde{\overline{\mathcal{P}}}^{\alpha\beta}}])\Tilde{z}^{\alpha}_{t+1}, \label{bellmen_blackuced_dp}
\end{align}
\end{theorem}
\begin{proof} 
See proof in Appendix \ref{appendix_proof_theorems}. 
\end{proof}
\allowdisplaybreaks
We note that the optimal control policy in \eqref{optimal_policy_dp} achieves differential privacy because the stochastic LS-MDP internalizes the Dirichlet noise, thus privatizing the optimization itself and making its output  private \cite{Dwork_book}. 
Notably, the outcome of Theorem~\ref{theorem_a} depends on the evaluation of $\mathbb{E}_{\Tilde{\overline{\mathcal{P}}}^{\alpha\beta}}\! [\log{\Tilde{\overline{\mathcal{P}}}^{\alpha\beta}}]$. We, therefore, propose two methods for computing  $\mathbb{E}_{\Tilde{\overline{\mathcal{P}}}^{\alpha\beta}}\! [\log{\Tilde{\overline{\mathcal{P}}}^{\alpha\beta}}]$ in Theorem \ref{theorem_a} under a given Dirichlet distribution\footnote{The term $\mathbb{E}_{\Tilde{\overline{\mathcal{P}}}^{\alpha\beta}}\! [\log{\Tilde{\overline{\mathcal{P}}}^{\alpha\beta}}]$ is generic and other distributions such as Normal and Laplace distributions can be used to further evaluate it. However, the output of non-Dirichlet distributions (where output does not belong to the unit simplex) needs to be adjusted or normalized to ensure that every row in the transition probability matrix remains equal to one. In other words, the probability of moving from present state $\beta$ to all possible next states $\alpha$ is equal to one.} and, thus, for enabling its computational tractable solution. As described below, the first method is based on the  Taylor approximation, while the second one is based on Digamma functions and yields an exact equivalent.  

\subsection{Taylor Approximation} \label{Taylor_Approximation}
We approximate $\mathbb{E}_{\Tilde{\overline{\mathcal{P}}}^{\alpha\beta}}\! [\log{\Tilde{\overline{\mathcal{P}}}^{\alpha\beta}}]$ by the second-order Taylor expansion to represent this term as a function of the first and second statistical moments (mean and variance): 
\begin{align}
\mathbb{E}_{\Tilde{\overline{\mathcal{P}}}^{\alpha\beta}}\! [\log{\Tilde{\overline{\mathcal{P}}}^{\alpha\beta}}] &\approx \log\mathbb{E}_{\Tilde{\overline{\mathcal{P}}}^{\alpha\beta}}\! [{\Tilde{\overline{\mathcal{P}}}^{\alpha\beta}}] - \frac{\text{Var}(\Tilde{\overline{\mathcal{P}}}^{\alpha\beta})}{2(\mathbb{E}_{\Tilde{\overline{\mathcal{P}}}^{\alpha\beta}}\! [{\Tilde{\overline{\mathcal{P}}}^{\alpha\beta}}])^2}, \label{taylor_app_1}
\end{align}
where $\text{Var}(\cdot)$ denotes variance.  Next, using the notion of the Dirichlet distribution, we invoke for the mean $(\mathbb{E}_{\Tilde{\overline{\mathcal{P}}}^{\alpha\beta}}\! [{\Tilde{\overline{\mathcal{P}}}^{\alpha\beta}}])$ in Eq.~\eqref{taylor_app_1}:
\begin{align}
\begin{split}
& \mathbb{E}_{\Tilde{\overline{\mathcal{P}}}^{\alpha\beta}}\! [{\Tilde{\overline{\mathcal{P}}}^{\alpha\beta}}] \!= \frac{k{\overline{\mathcal{P}}}^{\alpha\beta}}{\sum\limits_{\alpha\in\mathcal{A}}\!\!k\overline{\mathcal{P}}^{\alpha \beta}}\! = \frac{k{\overline{\mathcal{P}}}^{\alpha\beta}}{k\!\!\!\sum\limits_{\alpha\in\mathcal{A}}\!\!\overline{\mathcal{P}}^{\alpha \beta}}\! = \frac{k{\overline{\mathcal{P}}}^{\alpha\beta}}{k(1)} \!= \overline{\mathcal{P}}^{\alpha\beta}.
\end{split} \label{mean_taylor}
\end{align}
Similarly, using the notion of variance for the Dirichlet distribution, we obtain for the variance term in Eq.~\eqref{taylor_app_1}:
\begin{align}
\begin{split}
& \text{Var}(\Tilde{\overline{\mathcal{P}}}^{\alpha\beta}) = \frac{k\overline{\mathcal{P}}^{\alpha\beta}(\sum\limits_{\alpha\in\mathcal{A}}k\overline{\mathcal{P}}^{\alpha \beta}-k\overline{\mathcal{P}}^{\alpha\beta})}{(\sum\limits_{\alpha\in\mathcal{A}}k\overline{\mathcal{P}}^{\alpha \beta})^{2}(\sum\limits_{\alpha\in\mathcal{A}}k\overline{\mathcal{P}}^{\alpha \beta}+1)} = \\& \frac{k\overline{\mathcal{P}}^{\alpha\beta}(k\!\!\sum\limits_{\alpha\in\mathcal{A}}\!\!\overline{\mathcal{P}}^{\alpha \beta}-k\overline{\mathcal{P}}^{\alpha\beta})}{(k\!\!\sum\limits_{\alpha\in\mathcal{A}}\!\!\overline{\mathcal{P}}^{\alpha \beta})^{2}(k\!\!\sum\limits_{\alpha\in\mathcal{A}}\!\!\overline{\mathcal{P}}^{\alpha \beta}+1)} = \frac{\overline{\mathcal{P}}^{\alpha\beta}(1-\overline{\mathcal{P}}^{\alpha\beta})}{k+1}
\end{split} \label{variance_taylor}
\end{align}
Given \eqref{mean_taylor} and \eqref{variance_taylor}, the resulting value of Eq.~\eqref{taylor_app_1} can be approximated as:
\begin{align}
\begin{split}
&\mathbb{E}_{\Tilde{\overline{\mathcal{P}}}^{\alpha\beta}}\! [\log{\Tilde{\overline{\mathcal{P}}}^{\alpha\beta}}] \approx \log \overline{\mathcal{P}}^{\alpha\beta} - \frac{1-\overline{\mathcal{P}}^{\alpha\beta}}{2(\overline{\mathcal{P}}^{\alpha\beta})(k+1)}
\end{split} \label{expected_Taylor}
\end{align}
Given the DP LS-MDP formulation in \eqref{diff_priv_mdp} and expression in \eqref{expected_Taylor}, we propose: 

\begin{proposition} \label{theorem_b} \normalfont
Let \eqref{diff_priv_mdp} model a load ensemble as a differential private LS-MDP and let $\mathbb{E}_{\Tilde{\overline{\mathcal{P}}}^{\alpha\beta}}\! [\log{\Tilde{\overline{\mathcal{P}}}^{\alpha\beta}}]$ be approximated by the second-order Taylor approximation as in \eqref{expected_Taylor}. Then, the differentially private optimal control policy $\forall \beta,\alpha\in\mathcal{A}, \forall t\in\mathcal{T},$ is computed as:
\begin{align}
\begin{split}
&\Tilde{\mathcal{P}}_{t}^{\alpha \beta} = \frac{\overline{\mathcal{P}}^{\alpha\beta} \exp \Big[\frac{-1+\overline{\mathcal{P}}^{\alpha\beta}}{2 \overline{\mathcal{P}}^{\alpha\beta}(k+1)}\Big] \Tilde{z}^{\alpha}_{t+1}}{\sum\limits_{\alpha\in\mathcal{A}}\overline{\mathcal{P}}^{\alpha\beta} \exp \Big[\frac{-1+\overline{\mathcal{P}}^{\alpha\beta}}{2 \overline{\mathcal{P}}^{\alpha\beta}(k+1)}\Big] \Tilde{z}^{\alpha}_{t+1}}, \label{optimal_policy_dp_Taylor}
\end{split}
\end{align}
where $\forall \beta\in\mathcal{A}, \forall t\in\mathcal{T}$ and:
\begin{align}
&\Tilde{z}^{\beta}_{t} = \text{exp}\Big(\frac{U_{t}^{\beta}}{\gamma}\Big) \sum_{\alpha\in\mathcal{A}} \overline{\mathcal{P}}^{\alpha\beta} \exp\bigg[\frac{-1+\overline{\mathcal{P}}^{\alpha\beta}}{2\overline{\mathcal{P}}^{\alpha\beta}(k+1)}\bigg] \Tilde{z}^{\alpha}_{t+1}, \label{bellmen_blackuced_dp_Taylor}
\end{align}
\end{proposition}
\begin{proof} 
Our proof follows immediately from Theorem~\ref{theorem_a}. By using Eq.~\eqref{expected_Taylor} in Eqs.~\eqref{optimal_policy_dp} and \eqref{bellmen_blackuced_dp}, we end up with the results in Eqs.~\eqref{optimal_policy_dp_Taylor} and \eqref{bellmen_blackuced_dp_Taylor} and conclude the proof.
\end{proof}


\textcolor{black}{Since \eqref{diff_priv_mdp} handles the noise from Dirichlet distribution introduced for DP, its outcome is more costly as compared to the non-private optimal policy from Theorem~\ref{theorem_0}.} The difference in cost between the private and non-private policies, therefore, constitutes the cost of privacy that the decision-maker will pay for privatizing its data. Notably, the cost of privacy incured by \eqref{diff_priv_mdp} can be evaluated a priori as follows:

\begin{corollary} \label{cost_stochastic_taylor} \normalfont
Given the differentially private optimal control policy in \eqref{optimal_policy_dp_Taylor}, the cost of privacy for the stochastic approach using the second-order Taylor approximation is calculated as:
\begin{align} 
\begin{split}
&\Delta C \!= \gamma\!\! \sum_{\alpha \in \mathcal{A}}\!\! \frac{\overline{\mathcal{P}}^{\alpha\beta} \!\exp \Big[\frac{-1+\overline{\mathcal{P}}^{\alpha\beta}}{2 \overline{\mathcal{P}}^{\alpha\beta}(k+1)}\Big] \Tilde{z}^{\alpha}_{t+1}}{\sum\limits_{\alpha\in\mathcal{A}}\!\!\overline{\mathcal{P}}^{\alpha\beta}\! \exp \Big[\frac{-1+\overline{\mathcal{P}}^{\alpha\beta}}{2 \overline{\mathcal{P}}^{\alpha\beta}\!(k\!+\!1)}\Big] \Tilde{z}^{\alpha}_{t+1}} \bigg[\frac{-1+\overline{\mathcal{P}}^{\alpha\beta}}{2\overline{\mathcal{P}}^{\alpha\beta}(k+1)}\bigg] \\&\!\! + \!\gamma \log\!\! \sum_{\alpha \in \mathcal{A}}\!\overline{\mathcal{P}}^{\alpha \beta}z_{t+1}^{\alpha} - \gamma \!\log\!\! \sum_{\alpha\in\mathcal{A}}\! \overline{\mathcal{P}}^{\alpha\beta} \!\!\exp\! \Big[\frac{-1+\overline{\mathcal{P}}^{\alpha\beta}}{2 \overline{\mathcal{P}}^{\alpha\beta}\!(k\!+\!1)}\Big] \Tilde{z}^{\alpha}_{t+1}
\end{split} \label{eq:cost_privacy_stochastic_taylor}
\end{align}
\end{corollary}
\begin{proof} 
See proof in Appendix \ref{appendix_cost_stochastic}.
\end{proof}
While the result in \eqref{eq:cost_privacy_stochastic_taylor} is unwieldy, it is valuable for evaluating the effect of the Dirichlet distribution on the operating cost prior to implementing the noise for ensuring differentially private solutions. Notably, this method requires only a characterization of the first and second statistical moments, which are often available in practice \cite{Kotz_distribution}, and tuning of parameter $k$.

\subsection{Digamma Function Equivalent} \label{Digamma_function}
While using the Taylor approximation in Section~\ref{Section_Stochastic_Approach}\ref{Taylor_Approximation} provides a reasonably accurate approximation for  $\mathbb{E}_{\Tilde{\overline{\mathcal{P}}}^{\alpha\beta}}\! [\log{\Tilde{\overline{\mathcal{P}}}^{\alpha\beta}}]$, the decision-maker can alternatively derive analytical results for the differentially private control policy by utilizing the notion of Digamma functions \cite{Digamma_func} to represent the value of $\mathbb{E}_{\Tilde{\overline{\mathcal{P}}}^{\alpha\beta}}\![\log{\Tilde{\overline{\mathcal{P}}}^{\alpha\beta}}\!]$. However, the exactness of this policy comes at the expense of more complex representations that somewhat obstruct intuitive interpretations of the results. Accordingly, the use of Digamma functions leads to, \cite{Digamma_func}:
\begin{align}
\begin{split}
&\!\!\!\!\mathbb{E}_{\Tilde{\overline{\mathcal{P}}}^{\alpha\beta}}\![\log\!{\Tilde{\overline{\mathcal{P}}}^{\alpha\beta}}\!]\!=\!\psi(k\overline{\mathcal{P}}^{\alpha\beta}\!)\!-\!\psi(k\!\! \!\sum_{\alpha \in \mathcal{A}}\!\!\overline{\mathcal{P}}^{\alpha\beta})
\!= \!\psi(\!k\overline{\mathcal{P}}^{\alpha\beta}\!)\! - \!\psi(\!k\!) 
\end{split} \label{expected_Digamma}
\end{align}
where $\psi(.)$ represents the Digamma function, defined as the logarithmic derivative of the following form:
\begin{align}
& \psi(x) = \frac{d}{dx}\log[\Gamma(x)] = \frac{\Gamma^{\prime}(x)}{\Gamma(x)}
\end{align}
where $\Gamma(.)$ is the gamma function. Given the expression in \eqref{expected_Digamma}, we prove for the DP LS-MDP formulation in \eqref{diff_priv_mdp}: 

\begin{proposition} \label{theorem_c} \normalfont
Let \eqref{diff_priv_mdp} model a load ensemble as a differentially private LS-MDP and let $\mathbb{E}_{\Tilde{\overline{\mathcal{P}}}^{\alpha\beta}}\! [\log{\Tilde{\overline{\mathcal{P}}}^{\alpha\beta}}]$ be evaluated using the Digamma function as in \eqref{expected_Digamma}. Then, the differentially private optimal control policy $\forall \beta,\alpha\!\in\!\!\mathcal{A}, \forall t\!\in\!\mathcal{T},$ is computed as:
\begin{align}
\begin{split}
\!\!\!\!\!&\Tilde{\mathcal{P}}_{t}^{\alpha \beta} = \frac{\exp\!\!\Big[\psi(k\overline{\mathcal{P}}^{\alpha\beta})-\psi(k)\Big] \Tilde{z}^{\alpha}_{t+1}}{\sum\limits_{\alpha\in\mathcal{A}}\!\exp\!\!\Big[\psi(k\overline{\mathcal{P}}^{\alpha\beta})-\psi(k)\Big] \Tilde{z}^{\alpha}_{t+1}}, \label{optimal_policy_dp_Digamma}
\end{split}
\end{align}
where $\forall \beta\in\mathcal{A}, \forall t\in\mathcal{T}$ and:
\begin{align}
&\Tilde{z}^{\beta}_{t}\! = \text{exp}\Big(\frac{U_{t}^{\beta}}{\gamma}\Big)\!\! \sum_{\alpha\in\mathcal{A}}\!\! \exp\!\!\Big[\psi(k\overline{\mathcal{P}}^{\alpha\beta})-\psi(k)\Big] \Tilde{z}^{\alpha}_{t+1} \label{bellmen_blackuced_dp_Digamma}
\end{align}
\end{proposition}
\begin{proof}
Similarly to Proposition~\ref{theorem_b}, this proof follows immediately from Theorem~\ref{theorem_a}. By using Eq.~\eqref{expected_Digamma} in Eqs.~\eqref{optimal_policy_dp} and \eqref{bellmen_blackuced_dp}, we end up with the results in Eqs.~\eqref{optimal_policy_dp_Digamma} and \eqref{bellmen_blackuced_dp_Digamma} and conclude the proof.
\end{proof}

Similarly to Proposition \ref{theorem_b}, this optimal control policy in Proposition~\ref{theorem_c} internalizes the Dirichlet noise, but in a more accurate fashion since Digamma function provides an equivalent expression for $\mathbb{E}_{\Tilde{\overline{\mathcal{P}}}^{\alpha\beta}}\! [\log{\Tilde{\overline{\mathcal{P}}}^{\alpha\beta}}]$. This enhancement in accuracy of the optimal control policy, relative to the Taylor approximation above, results in added privacy, which, in turn, increases the cost of privacy as our later case study shows. Furthermore, the result of Proposition \ref{theorem_c} only requires tuning of parameter $k$ and does not necessitate any knowledge on the first- and second statistical moments. Using the result of Proposition~\ref{theorem_c}, the cost of privacy can then be computed as:

\begin{corollary} \label{cost_stochastic_digamma} \normalfont
Given the differential private optimal control policy in \eqref{optimal_policy_dp_Digamma}, the cost of privacy for the stochastic approach using the Digamma function equivalent is calculated as:
\begin{align}
\begin{split}
&\Delta C  = \gamma\!\! \sum_{\alpha \in \mathcal{A}}\! \frac{\exp\!\!\Big[\psi(k\overline{\mathcal{P}}^{\alpha\beta})-\psi(k)\Big] \Tilde{z}^{\alpha}_{t+1}}{\sum\limits_{\alpha\in\mathcal{A}}\!\exp\!\!\Big[\psi(k\overline{\mathcal{P}}^{\alpha\beta})-\psi(k)\Big] \Tilde{z}^{\alpha}_{t+1}} \Big[ \\& \psi(k\overline{\mathcal{P}}^{\alpha\beta})-\psi(k)\Big]+ \gamma \log \sum_{\alpha \in \mathcal{A}}\overline{\mathcal{P}}^{\alpha \beta}z_{t+1}^{\alpha} - \\& \gamma\!\! \sum_{\alpha \in \mathcal{A}}\! \frac{\exp\!\!\Big[\psi(k\overline{\mathcal{P}}^{\alpha\beta})-\psi(k)\Big] \Tilde{z}^{\alpha}_{t+1}}{\sum\limits_{\alpha\in\mathcal{A}}\!\exp\!\!\Big[\psi(k\overline{\mathcal{P}}^{\alpha\beta})-\psi(k)\Big] \Tilde{z}^{\alpha}_{t+1}} \log {\overline{\mathcal{P}}^{\alpha\beta}} \\& - \gamma \log \sum_{\alpha\in\mathcal{A}}\exp\Big[\psi(k\overline{\mathcal{P}}^{\alpha\beta})-\psi(k)\Big]\Tilde{z}^{\alpha}_{t+1}
\end{split} \label{cost_of_privacy_digamma}
\end{align}
\end{corollary}
\begin{proof} 
See proof in Appendix \ref{appendix_cost_stochastic}.
\end{proof}

Similarly to Corollary \ref{cost_stochastic_taylor}, the cost of privacy in \eqref{cost_of_privacy_digamma} depends on the default transition probabilities $\overline{\mathcal{P}}^{\alpha\beta}$, parameter $k$ of the Dirichlet mechanism, penalty parameter $\gamma$, and the value function represented as the desirability function $\Tilde{z}_{t+1}^{\alpha}$. \textcolor{black}{Furthermore, we note that the cost of privacy for the Digamma function equivalent, represented by \eqref{cost_of_privacy_digamma}, is expected to be more than the one calculated with the Taylor approximation in \eqref{eq:cost_privacy_stochastic_taylor}. Note that this because the Digamma function equivalent provides an exact equivalent for $\mathbb{E}_{\Tilde{\overline{\mathcal{P}}}^{\alpha\beta}}\! [\log{\Tilde{\overline{\mathcal{P}}}^{\alpha\beta}}]$ resulting in more privacy, while the Taylor approximation upper bounds the actual value of this term.}

\section{Average Value Approach} \label{Section_Avg_Value_Approach}
While the stochastic formulations described in Section~\ref{Section_Stochastic_Approach} provides useful analytical insights in the effects of differential privacy on the optimal control policy, these formulations are limited by the need to specify either statistical moments or invoke the notion of the digamma function, which requires differentiability. In practice, however, these assumptions can be restrictive and require the knowledge of ground-truth  default transition probabilities.
Therefore, we use the average value approach \cite{Sample_avg1,Sample_avg2} that does not require  knowing actual default transition probabilities and can  establish privacy at the aggregation level.  Thus, in the average value approach, we generate samples to obtain a set of $N$ differentially private default transition probabilities $\{\Tilde{\overline{\mathcal{P}}}^{\alpha\beta}_j\}_{j\in N},$ where $j$ is a sample index and each sample $\Tilde{\overline{\mathcal{P}}}^{\alpha\beta}_j$ in set $N$ is a realization of differentially private default transition probabilities $\Tilde{\overline{\mathcal{P}}}^{\alpha\beta}$. These samples are generated via the Dirichlet mechanism with default transition probabilities $\overline{\mathcal{P}}^{\alpha\beta}$ as input to the mechanism and samples as the output.
Using the generated samples, we then solve the optimization problem in \eqref{base_mdp} and obtain the expected differential private optimal control policy by averaging the optimal solutions obtained for each sample. That is, as given by Theorem~\ref{theorem_0}, for each sample we obtain the following optimal solution from \eqref{base_mdp}: 
\allowdisplaybreaks
\begin{equation}
\Tilde{\mathcal{P}}_{t,j}^{\alpha \beta} = \frac{\Tilde{\overline{\mathcal{P}}}^{\alpha\beta}_{j}\Tilde{z}^{\alpha}_{t+1,j}}{\sum\limits_{\alpha\in\mathcal{A}}\!\Tilde{\overline{\mathcal{P}}}_{j}^{\alpha\beta}\Tilde{z}^{\alpha}_{t+1,j}}, \label{optimal_policy_dp_avg}
\end{equation}
The optimal control policy in \eqref{optimal_policy_dp_avg} is differentially private since we feed the LS-MDP with privatized default transition probabilities $\Tilde{\overline{\mathcal{P}}}^{\alpha\beta}$ and the property that DP is immune to post-processing \cite{Dwork_book}, see Definition \ref{post_processing}.
The expected differential private optimal control policy is calculated as:
\begin{align}
& \mathbb{E}_{\Tilde{\overline{\mathcal{P}}}^{\alpha\beta}}[\Tilde{\mathcal{P}}_{t}^{\alpha \beta}] \approx \frac{1}{N}\sum_{j=1}^{N} \Tilde{\mathcal{P}}_{t,j}^{\alpha \beta} \label{average_policy_samples}
\end{align}
The expected policy obtained in \eqref{average_policy_samples} requires very large number of samples $N$ to obtain accurate results which could be time consuming for MDP problems. Therefore, we derive an expression for the expected differential private optimal control policy and propose: 

\begin{proposition} \label{average_expected_policy} \normalfont
Given a differential private optimal control policy in \eqref{optimal_policy_dp_avg} for each sample, we derive its expected value as:
\begin{align}
\begin{split}
& \mathbb{E}_{\Tilde{\overline{\mathcal{P}}}^{\alpha\beta}}[\Tilde{\mathcal{P}}_{t}^{\alpha \beta}] = \mathbb{E}_{\Tilde{\overline{\mathcal{P}}}^{\alpha\beta}}\Bigg[ \frac{\Tilde{\overline{\mathcal{P}}}^{\alpha\beta}\Tilde{z}^{\alpha}_{t+1}}{\sum\limits_{\alpha\in\mathcal{A}}\!\Tilde{\overline{\mathcal{P}}}^{\alpha\beta}\Tilde{z}^{\alpha}_{t+1}}\Bigg] \!\! \approx \frac{\overline{\mathcal{P}}^{\alpha\beta}\Tilde{z}^{\alpha}_{t+1}}{\sum\limits_{\alpha\in\mathcal{A}}\!\overline{\mathcal{P}}^{\alpha\beta}\Tilde{z}^{\alpha}_{t+1}} \\& - \frac{(\Tilde{z}^{\alpha}_{t+1})^2\overline{\mathcal{P}}^{\alpha\beta}(1-\overline{\mathcal{P}}^{\alpha\beta}) -\!\! \sum\limits_{\nu\neq\alpha\in\mathcal{A}}\!\!\Tilde{z}^{\alpha}_{t+1}\Tilde{z}^{\nu}_{t+1}\overline{\mathcal{P}}^{\alpha\beta}\overline{\mathcal{P}}^{\nu\beta}}{(k+1)\sum\limits_{\alpha\in\mathcal{A}}(\overline{\mathcal{P}}^{\alpha\beta}\Tilde{z}^{\alpha}_{t+1})^2} \\& + \frac{\splitfrac{\overline{\mathcal{P}}^{\alpha\beta}\Tilde{z}^{\alpha}_{t+1}\Big\{\sum\limits_{\alpha\in\mathcal{A}}(\Tilde{z}^{\alpha}_{t+1})^2\overline{\mathcal{P}}^{\alpha\beta}(1-\overline{\mathcal{P}}^{\alpha\beta})-\sum\limits_{\alpha\in\mathcal{A}}}{\sum\limits_{\nu
\neq\alpha\in\mathcal{A}}\Tilde{z}^{\alpha}_{t+1}\Tilde{z}^{\nu}_{t+1}\overline{\mathcal{P}}^{\alpha\beta}\overline{\mathcal{P}}^{\nu\beta}\Big\}}}{(k+1)\sum\limits_{\alpha\in\mathcal{A}}(\overline{\mathcal{P}}^{\alpha\beta}\Tilde{z}^{\alpha}_{t+1})^3}
\end{split} \label{expected_policy_average}
\end{align}
\end{proposition}
\begin{proof} 
See proof in Appendix \ref{appendix_proof_expected_policy}.
\end{proof}

Accordingly, we demonstrate the expected cost of privacy for the average value approach as follows:
\begin{corollary} \label{average_expected_cost} \normalfont
Given a differential private optimal control policy in \eqref{optimal_policy_dp_avg}, the expected cost of privacy for the average value approach is calculated as:
\begin{align}
\begin{split}
&\Delta C \approx \gamma\sum\limits_{\alpha\in\mathcal{A}}\Bigg[\frac{(\Tilde{z}^{\alpha}_{t+1})^2\overline{\mathcal{P}}^{\alpha\beta}(1-\overline{\mathcal{P}}^{\alpha\beta})}{(k+1)\sum\limits_{\alpha\in\mathcal{A}}(\overline{\mathcal{P}}^{\alpha\beta}\Tilde{z}^{\alpha}_{t+1})^2} \\& - \frac{\splitfrac{2\overline{\mathcal{P}}^{\alpha\beta}\Tilde{z}^{\alpha}_{t+1}\Big\{(\Tilde{z}^{\alpha}_{t+1})^2\overline{\mathcal{P}}^{\alpha\beta}(1-\overline{\mathcal{P}}^{\alpha\beta})-}{\sum\limits_{\nu
\neq\alpha\in\mathcal{A}}\Tilde{z}^{\alpha}_{t+1}\Tilde{z}^{\nu}_{t+1}\overline{\mathcal{P}}^{\alpha\beta}\overline{\mathcal{P}}^{\nu\beta}\Big\}}}{(k+1)\sum\limits_{\alpha\in\mathcal{A}}(\overline{\mathcal{P}}^{\alpha\beta}\Tilde{z}^{\alpha}_{t+1})^3} \\& + \frac{\splitfrac{(\overline{\mathcal{P}}^{\alpha\beta}\Tilde{z}^{\alpha}_{t+1})^2\Big\{\sum\limits_{\alpha\in\mathcal{A}}(\Tilde{z}^{\alpha}_{t+1})^2\overline{\mathcal{P}}^{\alpha\beta}(1-\overline{\mathcal{P}}^{\alpha\beta})-}{\sum\limits_{\alpha\in\mathcal{A}}\sum\limits_{\nu
\neq\alpha\in\mathcal{A}}\Tilde{z}^{\alpha}_{t+1}\Tilde{z}^{\nu}_{t+1}\overline{\mathcal{P}}^{\alpha\beta}\overline{\mathcal{P}}^{\nu\beta}\Big\}}}{(k+1)\sum\limits_{\alpha\in\mathcal{A}}(\overline{\mathcal{P}}^{\alpha\beta}\Tilde{z}^{\alpha}_{t+1})^4} \\& - \mathbb{E}_{\Tilde{\overline{\mathcal{P}}}^{\alpha\beta}}[\Tilde{\mathcal{P}}_{t}^{\alpha\beta}](\log {\overline{\mathcal{P}}^{\alpha\beta}}\!\! - \Tilde{z}_{t+1}^{\alpha})\Bigg] + \gamma \log\!\! \sum_{\alpha \in \mathcal{A}}\!\overline{\mathcal{P}}^{\alpha \beta}z_{t+1}^{\alpha} \label{expected_cost_average}
\end{split}
\end{align}
where $\mathbb{E}_{\Tilde{\overline{\mathcal{P}}}^{\alpha\beta}}[\Tilde{\mathcal{P}}_{t}^{\alpha\beta}]$ is given by Eq.~\eqref{expected_policy_average}.
\end{corollary}
\begin{proof} 
See proof in Appendix \ref{appendix_proof_average_cost}.
\end{proof}

\section{Case Study} \label{Sec:Case Study}
The case study uses electricity power consumption data and available control means (e.g. HVAC, lightning, elevators) from one campus building of New York University (NYU)  located in Manhattan, NY, \cite{Ali_Hierarchical}. This data spans across 2018 and includes the participation in several DR events of the commercial system relief program (CSRP) operated by Consolidated Edison, the local electric power utility. Using this data, we generate an ensemble of 100 synthetic buildings, where energy consumption data is randomized by adding a random Gaussian noise to the original data such that the resulting synthetic consumption remains within $\pm 10\%$ of the original values. The aggregated power consumption data for this ensemble of 100 buildings is displayed in Fig.~\ref{fig:power_agg_plus_MP} (a) and is used to construct the MP with 20 Markovian states for the summer season (from mid-June to late-September) when all DR events in this region occurred in 2018. The resulting default transition probabilities $(\overline{\mathcal{P}}^{\alpha\beta}$) are shown in Fig.~\ref{fig:power_agg_plus_MP} (b). For simulating a realistic DR event, we isolate the period from 11:00 AM to 3:00 PM on July 2\textsuperscript{nd}, 2018 when such an event happened in Manhattan, NY \cite{DR_Events_Coned} and NYU buildings were called upon to provide DR. We also price electricity consumption during the DR event using real-world tariff from Consolidated Edison, \cite{DR_Events_Coned}, and set the discomfort penalty paid to individual DR participants  $\gamma=\$15$.

\begin{figure}[!t]
\centering 
\includegraphics[width=\columnwidth,trim={0 3.5cm 0 0}]{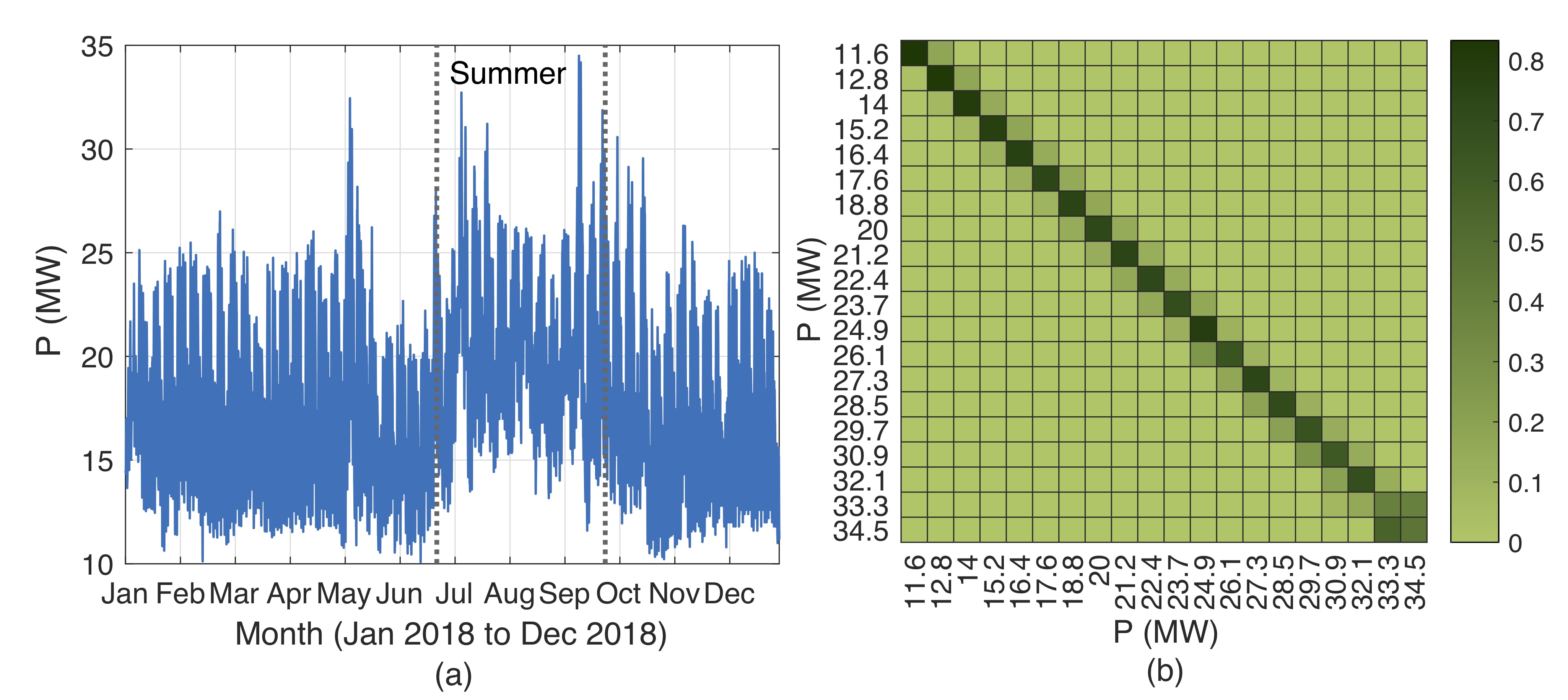}
\caption{(a) Aggregated power consumption of 100 buildings, (b) default transition probability matrix with 20 states constructed from the power profile (Summer) in (a), where color density indicates the probability value.}
\label{fig:power_agg_plus_MP}
\end{figure}

\textit{1) Effect of privatization:} We first illustrate the effect of privatization with the Dirichlet distribution on the the default transition probabilities ($\overline{\mathcal{P}}^{\alpha\beta}$) given in Fig.~\ref{fig:power_agg_plus_MP} (b). To this end, we take row vector 18 from the default transition probability matrix in Fig.~\ref{fig:power_agg_plus_MP} (b), which corresponds to the state with power consumption of 32.1 MW, as vector $\overline{\zeta}^{\beta}$ and obtain its adjacent vector $\overline{\eta}^{\beta}$ by using Definition~\ref{def:adjacent_vectors} and setting $h=0.03$. \textcolor{black}{This means we subtract 0.015 from one of the entries (entry 18 in our case) in vector $\overline{\zeta}^{\beta}$ and add 0.015 to one of the remaining entries (entry 19 in our case).} We set $k=50$ and $k=200$ to perform 1000 independent executions of the Dirichlet mechanism for non-zero elements of both $\overline{\zeta}^{\beta}$ and $\overline{\eta}^{\beta}$ and plot the output private default transition probabilities in Fig.~\ref{fig:privacy_guarantee}.  We note that, if $k=50$ and the maximum probability of
privacy failure is $\delta=0.05$, the mechanism is (0.88, 0.05)-differentially private as per Theorem~\ref{theorem_dm_privacy}, and it is impossible to distinguished between the outcomes, especially with a high probability, of  $\overline{\zeta}^{\beta}$ or $\overline{\eta}^{\beta}$. However, if $\delta=0.05$ for $k=200$, we found the mechanism is (5.45, 0.05)-differentially private, i.e. less private then in the case with $k=50$. In this case, the outcomes of the mechanism for $\overline{\zeta}^{\beta}$ and $\overline{\eta}^{\beta}$ are relatively distinguishable (see  Fig.~\ref{fig:privacy_guarantee}). Thus, we can achieve stronger privacy guarantees for lower values of $k$.

\begin{figure}[!t]
\centering 
\includegraphics[width=0.95\columnwidth,trim={0 2.5cm 0 0}]{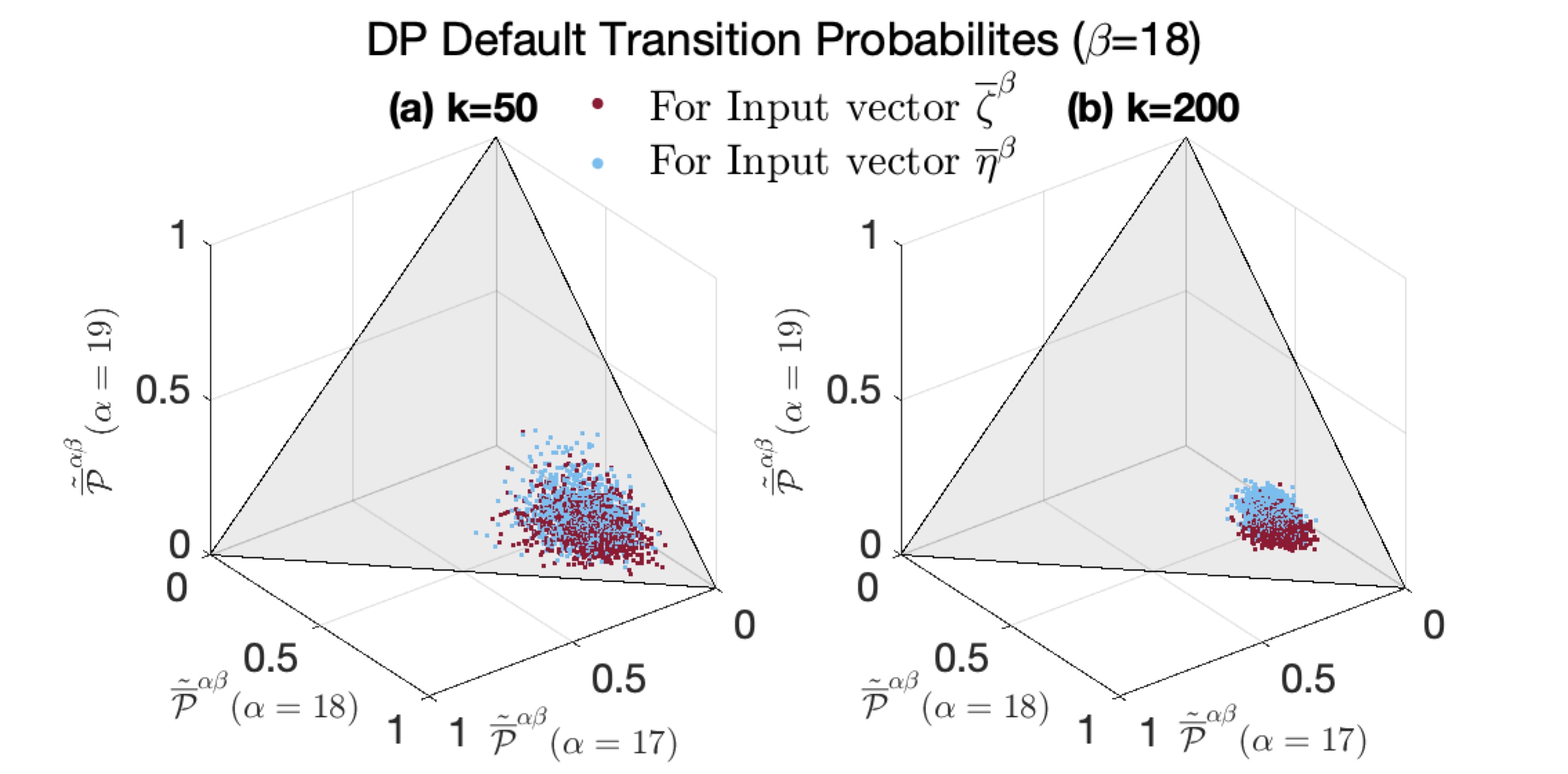}
\caption{Effects of privatization with the  Dirichlet distribution: (a) $k=50$, (0.88, 0.05)-differentially private,  (b) $k=200$, (5.45, 0.05)-differentially private. The shaded triangle defines the output space of the Dirichlet mechanism exhibiting integrality of a unit simplex.  Red and blue dots represent  independent runs for adjacent vectors $\overline{\zeta}^{\beta}$ and $\overline{\eta}^{\beta}$.}
\label{fig:privacy_guarantee}
\end{figure}

\textit{2) Effect on the optimal control policy:} Next, we implement the stochastic approach to obtain analytical differentially private optimal control policies. Fig.~\ref{fig:optimal_policy_stochastic} compares the resulting policies for the two adjacent input vectors $\overline{\zeta}^{\beta}$ and $\overline{\eta}^{\beta}$ with $h=0.03$. We plot the optimal control policy for the state $\beta=18$ at time 12:45 PM and show optimal transitions from $\beta$ to all other non-zero states $\alpha$. We note that the two adjacent inputs provide  similar control policies (as outputs) for both the Taylor approximation and Digamma equivalent methods. In fact, the $L_1$-norm distance between the two output vectors in Fig.~\ref{fig:optimal_policy_stochastic} is 0.0493 when $k=50$ and 0.0482 when $k=200$ for the Taylor approximation, and 0.0494 when $k=50$ and 0.0482 when $k=200$ for the Digamma equivalent. Similarly, based on private default transition probabilities $({\Tilde{\overline{\mathcal{P}}}}^{\alpha\beta})$ in Fig.~\ref{fig:privacy_guarantee}, we implement the average value approach to obtain private optimal control policies. The results of 1000 independent runs of the LS-MDP method for 1000 independent samples of DP default transition probabilities are shown in Fig.~\ref{fig:optimal_policy_average} (a) for $K=50$ and Fig.~\ref{fig:optimal_policy_average}  (b) for $k=200$, where red dots correspond to input vector $\overline{\eta}^{\beta}$ and blue dots to $\overline{\zeta}^{\beta}$. We computed the expected DP optimal policy for the average value approach and illustrated it in Fig.~\ref{fig:optimal_policy_average} (c) for $k=50$ and Fig.~\ref{fig:optimal_policy_average} (d) for $k=200$. The $L_1$-norm distance between the two outputs for $k=50$ is  0.0479 and for $k=200$ is  0.0476.

\begin{figure}[!t]
\centering 
\includegraphics[width=0.95\columnwidth,trim={0 3cm 0 0}]{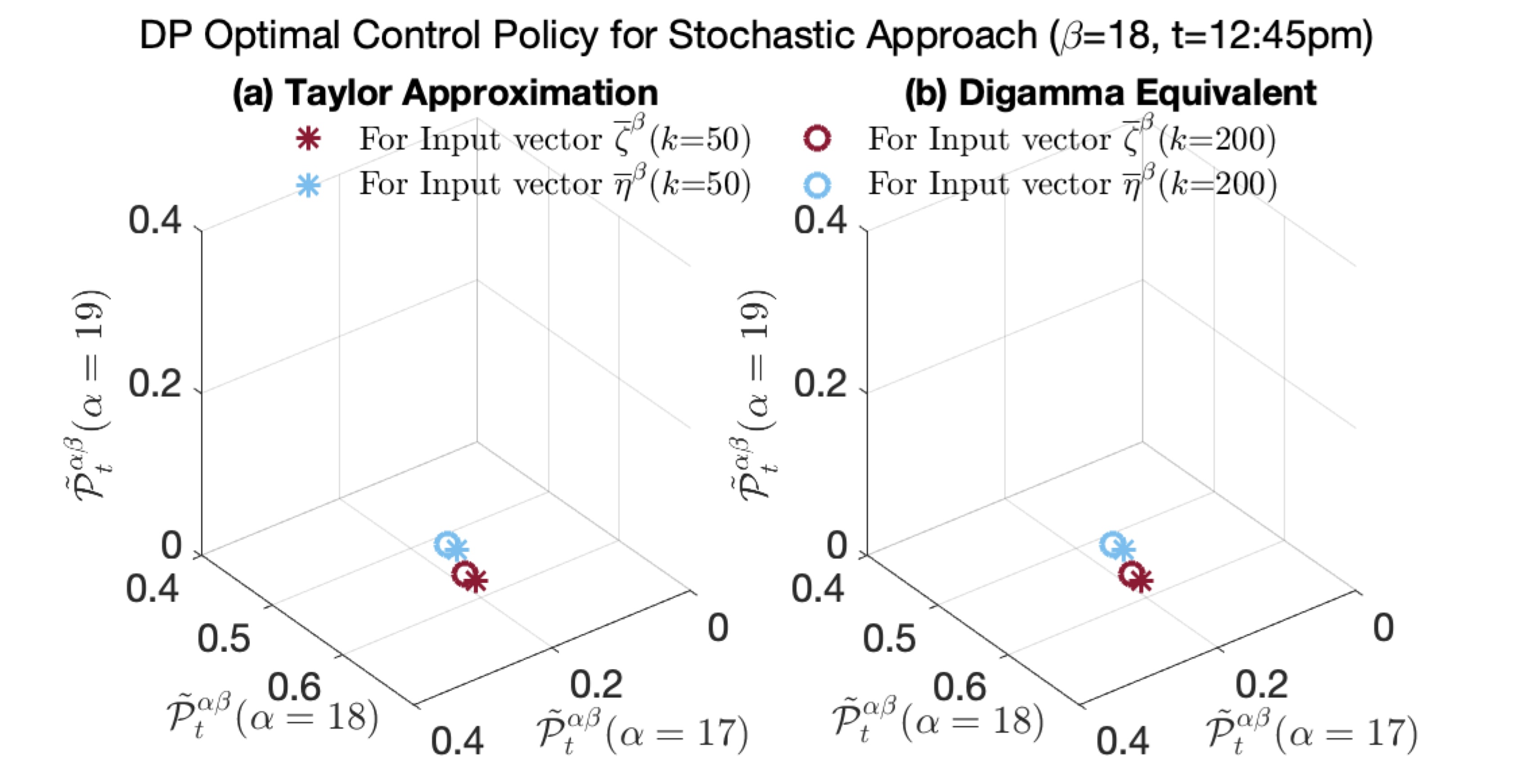}
\caption{Comparison of the DP  control policies for  adjacent  vectors ($h=0.03$) for (a) Taylor  and (b) Digamma  methods from   state $\beta=18$ to states 17, 18 and 19 at  12:45 pm.}
\label{fig:optimal_policy_stochastic}
\end{figure}

\begin{figure}[!t]
\centering 
\includegraphics[width=\columnwidth,trim={0 5cm 0 0}]{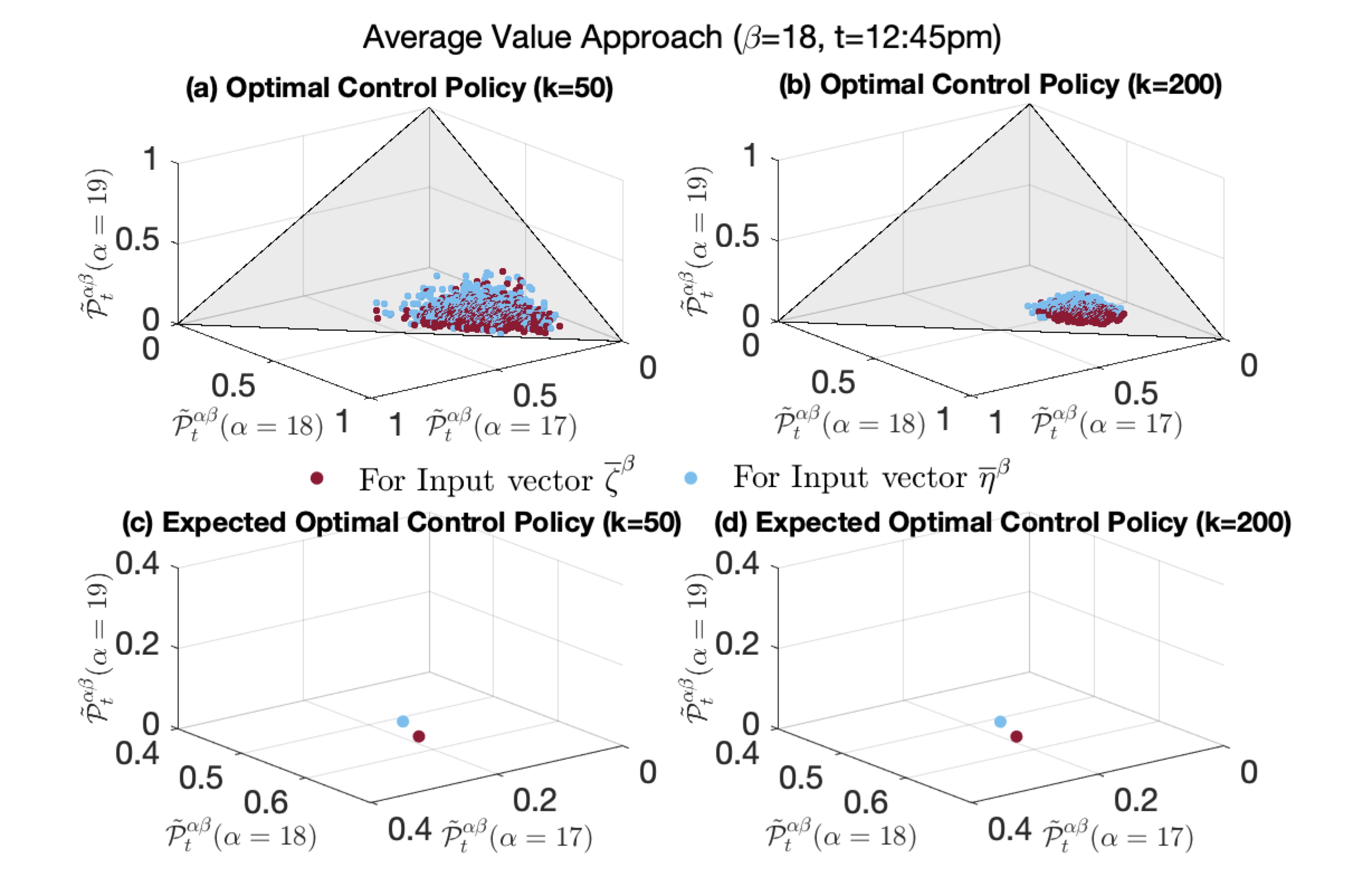}
\caption{Comparison of the DP optimal control policies for adjacent  vectors ($h=0.03$) for the average value method. The plots are the optimal transition probabilities from state $\beta=18$ to states 17, 18 and 19 at  12:45 pm. Each red and blue dot in (a) and (b) represents an independent run for inputs $\overline{\eta}^{\beta}$ and $\overline{\zeta}^{\beta}$ for $k=50$ and $k=200$, respectively, and the shaded triangle defines the output space of the Dirichlet mechanism exhibiting integrality of a unit simplex. The red and blue dots in (c) and (d) represent the expected optimal policy for both inputs $\overline{\eta}^{\beta}$ and $\overline{\zeta}^{\beta}$ for $k=50$ and $k=200$, respectively.}
\label{fig:optimal_policy_average}
\end{figure}

\textit{3) Effect of privacy on a DR event:} Using both the stochastic and average value approaches, we obtain  the optimal control policies for one DR event (11:00 AM to 3:00 PM on July 2\textsuperscript{nd}, 2018) in which the  ensemble of 100 NYU buildings  participated, and present our results in Fig.~\ref{fig:power_event}. Note that in each case we implement the DR protocol 30 minutes before the DR event to mirror the real-life conditions and without perfect foresight of the future. First, observe that all MDP-based policies extract more capacity from the ensemble than the ad-hoc current practice. Second, the DR curtailments are slightly decreased for privacy-aware solutions as expected due to the introduction of the Dirichlet noise. We also observe that both the Digamma and Taylor approaches lead to similar DR capacity extracted during the DR event with the least deviation from the non-privatized response obtained with the standard MDP optimization. On the other hand, the average value approach leads to the extraction of $\sim15\%$ less capacity at the peak of the DR event than the analytical policies obtained with the Digamma and Taylor approaches.

This performance difference between the stochastic and average value approaches is also observed in their privacy guarantees and costs, where generally more private policies yield less DR capacity extracted from the ensemble and greater privacy costs. We present the summary of privacy guarantees and privacy costs in Fig.~\ref{fig:privacy_cost_comp} for different values of  parameter $k$. Thus, among the analytical policies, the Digamma  method is slightly more expensive than the Taylor approximation, but provides more accurate solutions. On the other hand, the average value method is more expensive as compared to both the Taylor and Digamma methods. In general, privacy decreases for all methods with a decrease in the value of $k$, indicating a trade-off between solution privacy and cost as a function of the added Dirichlet noise.

\begin{figure}[!t]
\centering 
\includegraphics[width=0.9\columnwidth]{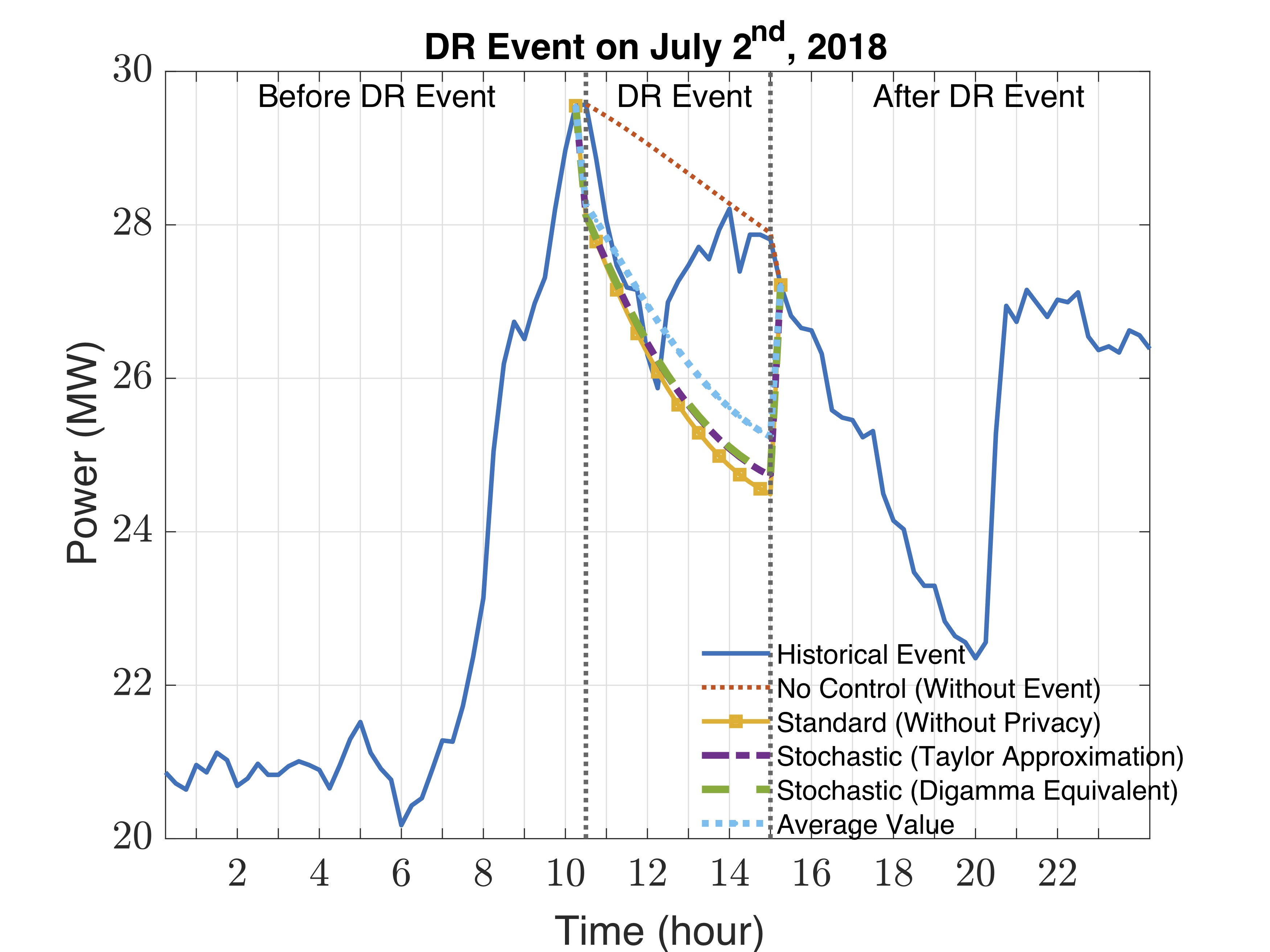}
\caption{Performance of the standard and private MDP methods for the DR event from 11:00 AM to 3:00 PM on July 2\textsuperscript{nd}, 2018.}
\label{fig:power_event}
\end{figure}

\begin{figure}[!t]
\centering 
\includegraphics[width=0.85\columnwidth]{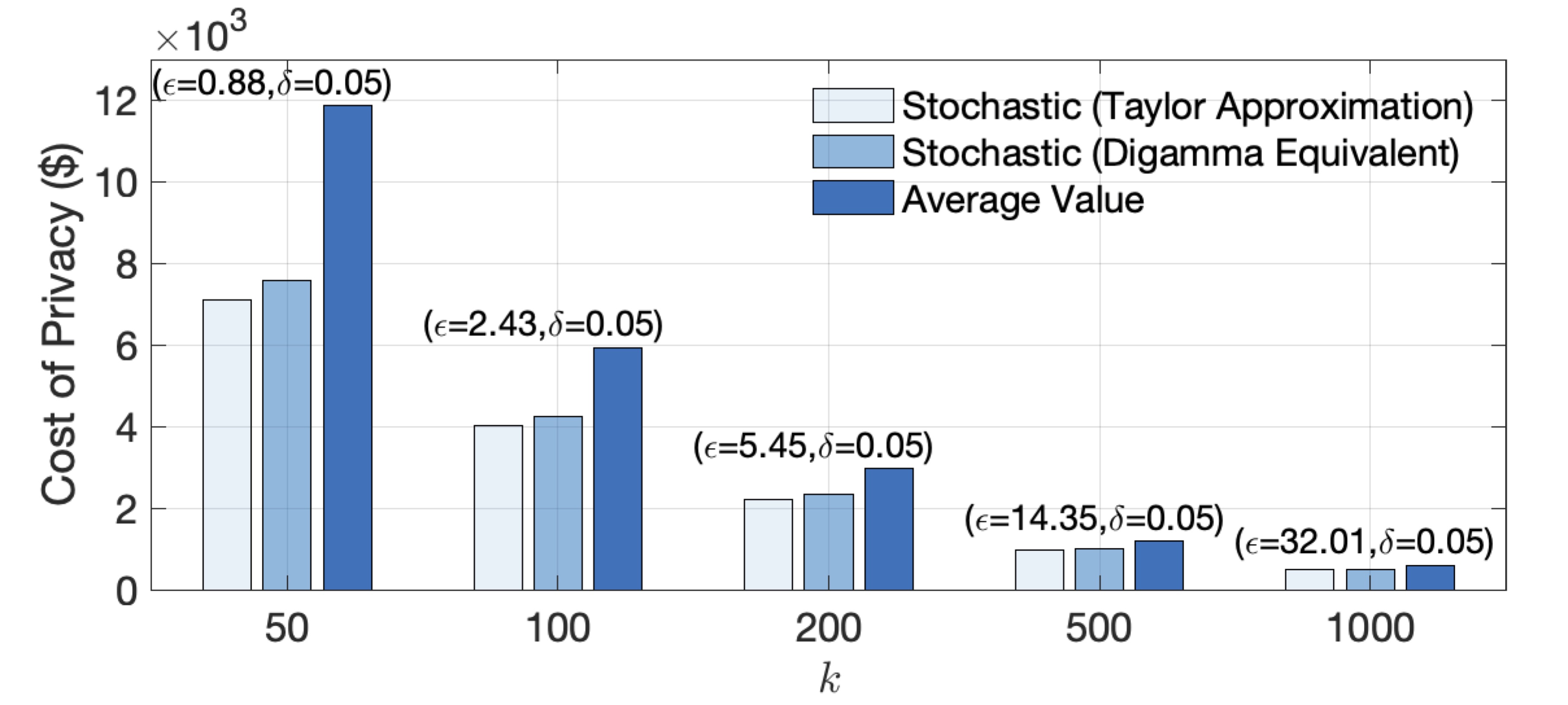}
\caption{Comparison of the privacy costs and  guarantees for the stochastic and average value methods for different values of  $k$, where $\delta = 0.05$ remains constant and $\epsilon$ changes as a function of $k$.}
\label{fig:privacy_cost_comp}
\end{figure}

\section{Conclusion}
This paper extended existing MDP formulations for dispatching ensembles of electrical loads to incorporate the notion of differential privacy that safeguards energy consumption data of DR participants. This is achieved by leveraging the Dirichlet mechanism, which preserves unit simplex (i.e., the sum of transition probabilities for a given state is equal to one), and its coupling with the two variances of the stochastic MDP formulations. The first variant -- stochastic approach --  privatizes the optimization routine of the aggregator, but has greater data requirements. The second variant -- average value approach --  reduces data constraints, but ensures differential privacy using its post-processing property.  The case study is based on real-world demand response data and demonstrates that both the stochastic and average value approaches provide privacy-cognizant solutions at the expense of a slight increase in operating costs,  reflected in the cost of privacy, but provide robust privacy guarantees and allow for trading-off between solution quality and privacy risks.

\bibliographystyle{IEEEtran}
\bibliography{ref.bib}


\appendices

\section{Proofs of Theorems \ref{theorem_0} and \ref{theorem_a}} \label{appendix_proof_theorems}
We follow the same procedure to prove Theorems \ref{theorem_0} and \ref{theorem_a} and denote  theorem-specific terms as $\mathcal{Z}^{\alpha\beta}$, where $\mathcal{Z}^{\alpha\beta}$ for each theorem is derived at the end of this appendix. We write the following Bellman equation $\forall t$ and $\forall \beta$:
\begin{align}
\begin{split}
&\hspace{-3mm} \frac{\varphi^{\beta}_{t}}{\gamma} \!=\! \frac{1}{\gamma} \underset{\substack{\mathcal{P}}}{\text{min}} \big(\!\!-\!U_{t}^{\beta}\!\! + \!\mathbb{E}_{\mathcal{P}^{\alpha\beta}_{t}} \! \big[\gamma\! \log \mathcal{P}_{t}^{\alpha\beta} \!-\!\gamma\!\log\mathcal{Z}^{\alpha\beta}\! + \varphi^{\alpha}_{t+1}\big]\! \big), \label{bellmen_1} 
\end{split} 
\end{align}
where $\varphi^{\beta}_{t}$ is the value function of the load ensemble at present state $\beta$ at time $t$ and $\varphi^{\alpha}_{t+1}$ is the value function at next state $\alpha$ at time $t+1$. By introducing desirability function $z^{\beta}_{t} = \text{exp}(\frac{-\varphi^{\beta}_{t}}{\gamma})$ in \eqref{bellmen_1} we obtain:
\begin{align}
\begin{split}
& -\!\!\text{log}z^{\beta}_{t}\! =\! 
\frac{1}{\gamma} \underset{\substack{\mathcal{P}}}{\text{min}} \bigg(\!\!-U_{t}^{\beta} \!+ \!\gamma \mathbb{E}_{\mathcal{P}^{\alpha\beta}_{t}}\! \bigg[ \!\log \frac{\mathcal{P}_{t}^{\alpha\beta}} {\mathcal{Z}^{\alpha\beta} z^{\alpha}_{t+1}} \bigg] \!\bigg). \label{bellmen_5} 
\end{split}
\end{align}
After normalizing as  $\mathcal{G}^{\beta}_{t}(z)=\sum\limits_{\alpha\in\mathcal{A}}\mathcal{Z}^{\alpha\beta}z^{\alpha}_{t+1}$, \eqref{bellmen_5} is recast as:
\begin{align}
\begin{split}
&\!\! \!-\!\text{log}(\!z^{\beta}_{t}\!)\! = \!
\frac{-U_{t}^{\beta}}{\gamma}\! + \underset{\substack{\mathcal{P}}}{\text{min}} 
KL\! \bigg[\!\mathcal{P}_{t}^{\alpha\beta} \bigg\Vert \frac{\mathcal{Z}^{\alpha\beta}z^{\alpha}_{t+1}}{\mathcal{G}^{\beta}_{t}(z)} \!\bigg] \!-\! \text{log}\mathcal{G}^{\beta}_{t}(\!z\!)\! \label{bellmen_9} 
\end{split}
\end{align}
By equating the two distributions in the KL divergence, i.e. setting $KL\big[\cdot|\cdot\big]=0$ , the optimal policy follows as:
\begin{align}
&\mathcal{P}_{t}^{\alpha \beta} = \frac{\mathcal{Z}^{\alpha\beta}z^{\alpha}_{t+1}}{\mathcal{G}^{\beta}_{t}(z)} = \frac{\mathcal{Z}^{\alpha\beta}z^{\alpha}_{t+1}}{\sum\limits_{\alpha\in\mathcal{A}}\mathcal{Z}^{\alpha\beta}z^{\alpha}_{t+1}}. \label{optimal_policy_1}
\end{align}
Then, this optimal policy converts \eqref{bellmen_9} to:
\begin{align}
\begin{split}
& -\text{log}(z^{\beta}_{t}) =
- \Big\{\frac{U_{t}^{\beta}}{\gamma} + \text{log} \Big[\sum_{\alpha}\mathcal{Z}^{\alpha\beta}z^{\alpha}_{t+1} \Big] \Big\}. \label{bellmen_blackuced_2} 
\end{split}
\end{align}
Exponentiating Eq.~\eqref{bellmen_blackuced_2} leads to the following  form:
\begin{align}
&z^{\beta}_{t} = \text{exp}\Big(\frac{U_{t}^{\beta}}{\gamma}\Big) \sum_{\alpha}\mathcal{Z}^{\alpha\beta}z^{\alpha}_{t+1}. \label{bellmen_blackuced_3}
\end{align}
Since the value of $\mathcal{Z}^{\alpha\beta}$ varies for Theorems \ref{theorem_0} and \ref{theorem_a}, we derive theorem-specific results for each case below:

\noindent \underline{For Theorem \ref{theorem_0}:} We have $\mathcal{Z}^{\alpha\beta} = \overline{\mathcal{P}}^{\alpha\beta}$, and using this value of $\mathcal{Z}^{\alpha\beta}$ in \eqref{optimal_policy_1} concludes the proof for Theorem \ref{theorem_0}.

\noindent \underline{For Theorem \ref{theorem_a}:} We have $\mathcal{Z}^{\alpha\beta} = \mathbb{E}_{\Tilde{\overline{\mathcal{P}}}^{\alpha\beta}} [\log {\Tilde{\overline{\mathcal{P}}}^{\alpha\beta}}]$, and using this value of $\mathcal{Z}^{\alpha\beta}$ in \eqref{optimal_policy_1} concludes the proof for Theorem \ref{theorem_a}.

\vspace{10mm}

\section{Proof of Corollaries \ref{cost_stochastic_taylor}, \ref{cost_stochastic_digamma}} \label{appendix_cost_stochastic}
Let $\varphi^{\beta}_{t}$ represents the value function for the non-private LS-MDP and $\Tilde{\varphi}^{\beta}_{t}$ for the DP LS-MDP. We start with the non-private LS-MDP and $\forall \beta\in\mathcal{A},t\in\mathcal{T}$:
\begin{align}
\begin{split}
& \varphi^{\beta}_{t} = -U_{t}^{\beta} +\!\! \sum_{\alpha \in \mathcal{A}}\!\mathcal{P}_{t}^{\alpha\beta}(\gamma \log \mathcal{P}_{t}^{\alpha\beta} - \gamma\log {\overline{\mathcal{P}}^{\alpha\beta}} + \varphi^{\alpha}_{t+1}) 
\end{split} \label{Val_func_det_cost}
\end{align}
By using \eqref{optimal_policy} from Theorem \ref{theorem_0} in \eqref{Val_func_det_cost}, we obtain:

\begin{align}
\begin{split}
& \varphi^{\beta}_{t} = -U_{t}^{\beta} - \gamma \log \sum_{\alpha \in \mathcal{A}}\overline{\mathcal{P}}^{\alpha \beta}z_{t+1}^{\alpha} 
\end{split}
\end{align}
Similarly, for the DP LS-MDP and $\forall \beta\in\mathcal{A},t\in\mathcal{T}$:
\begin{align}
\begin{split}
& \Tilde{\varphi}^{\beta}_{t} = -U_{t}^{\beta} + \!\!\sum_{\alpha \in \mathcal{A}}\!\Tilde{\mathcal{P}}_{t}^{\alpha\beta}( \gamma \log \Tilde{\mathcal{P}}_{t}^{\alpha\beta} - \gamma \log {\overline{\mathcal{P}}^{\alpha\beta}} + \Tilde{\varphi}^{\alpha}_{t+1}) 
\end{split} \label{Val_func_dp_cost}
\end{align}
By using \eqref{optimal_policy_dp} from Theorem \ref{theorem_a} in \eqref{Val_func_dp_cost}, we obtain:
\begin{align}
\begin{split}
&\!\!\Tilde{\varphi}^{\beta}_{t}\! = -U_{t}^{\beta} \!+\! \gamma\!\! \sum_{\alpha \in \mathcal{A}}\!\! \Tilde{\mathcal{P}}_{t}^{{\alpha \beta}} \mathbb{E}_{\Tilde{\overline{\mathcal{P}}}^{\alpha\beta}}\! [\log\! {\Tilde{\overline{\mathcal{P}}}^{\alpha\beta}}\!]\! - \gamma\!\! \sum_{\alpha \in \mathcal{A}}\!\! \Tilde{\mathcal{P}}_{t}^{{\alpha \beta}} \! \log \!{\overline{\mathcal{P}}^{\alpha\beta}} \\& - \gamma \log \sum_{\alpha\in\mathcal{A}}\exp(\mathbb{E}_{\Tilde{\overline{\mathcal{P}}}^{\alpha\beta}} [\log {\Tilde{\overline{\mathcal{P}}}^{\alpha\beta}}])\Tilde{z}^{\alpha}_{t+1}
\end{split}
\end{align}
We define the cost of privacy as the difference between the value functions of the non-private and DP LS-MDPs:
\begin{align}
& \Delta C  = \Delta \varphi_{t}^{\beta} = \Tilde{\varphi_{t}}^{\beta} - \varphi_{t}^{\beta} \nonumber \\
\begin{split}
& = \gamma\!\! \sum_{\alpha \in \mathcal{A}}\! \Tilde{\mathcal{P}}_{t}^{{\alpha \beta}} \mathbb{E}_{\Tilde{\overline{\mathcal{P}}}^{\alpha\beta}} [\log {\Tilde{\overline{\mathcal{P}}}^{\alpha\beta}}] - \gamma\!\! \sum_{\alpha \in \mathcal{A}}\! \Tilde{\mathcal{P}}_{t}^{\alpha \beta} \log {\overline{\mathcal{P}}^{\alpha\beta}} - \\& \gamma \log \!\!\sum_{\alpha\in\mathcal{A}}\!\exp(\mathbb{E}_{\Tilde{\overline{\mathcal{P}}}^{\alpha\beta}} [\log {\Tilde{\overline{\mathcal{P}}}^{\alpha\beta}}])\Tilde{z}^{\alpha}_{t+1} + \gamma \log\!\! \sum_{\alpha \in \mathcal{A}}\!\overline{\mathcal{P}}^{\alpha \beta}z_{t+1}^{\alpha}
\end{split} \label{cost_of_privacy_gen}
\end{align}

\noindent \underline{For Corollary \ref{cost_stochastic_taylor}:} Using \eqref{optimal_policy_dp_Taylor} in \eqref{cost_of_privacy_gen}, we conclude the proof. 

\noindent \underline{For Corollary \ref{cost_stochastic_digamma}:} Using \eqref{optimal_policy_dp_Digamma} in \eqref{cost_of_privacy_gen}, we conclude the proof. 

\vspace{10mm}
\section{Proof of Proposition \ref{average_expected_policy}} \label{appendix_proof_expected_policy}
Recall the multi-variate Taylor expansion:
\begin{align}
\begin{split}
& \mathbb{E}\bigg[\frac{X}{Y} \bigg] \approx \frac{\mathbb{E}[X]}{\mathbb{E}[Y]} - \frac{\text{Cov}[X,Y]}{\mathbb{E}[Y]^2} + \frac{\mathbb{E}[X]\text{Var}[Y]}{\mathbb{E}[Y]^3} \label{avg_optimal_policy_2}
\end{split}
\end{align}
where $\text{Cov}[X,Y]$ represents the covariance between X and Y. Using \eqref{avg_optimal_policy_2}, we proceed as:
\begin{align}
\begin{split}
& \mathbb{E}_{\Tilde{\overline{\mathcal{P}}}^{\alpha\beta}}\Bigg[ \frac{\Tilde{\overline{\mathcal{P}}}^{\alpha\beta}\Tilde{z}^{\alpha}_{t+1}}{\sum\limits_{\alpha\in\mathcal{A}}\!\Tilde{\overline{\mathcal{P}}}^{\alpha\beta}\Tilde{z}^{\alpha}_{t+1}}\Bigg] \approx \frac{\mathbb{E}_{\Tilde{\overline{\mathcal{P}}}^{\alpha\beta}}[\Tilde{\overline{\mathcal{P}}}^{\alpha\beta}]\Tilde{z}^{\alpha}_{t+1}}{\sum\limits_{\alpha\in\mathcal{A}}\!\mathbb{E}_{\Tilde{\overline{\mathcal{P}}}^{\alpha\beta}}[\Tilde{\overline{\mathcal{P}}}^{\alpha\beta}]\Tilde{z}^{\alpha}_{t+1}} \\& - \frac{\text{Cov}[\Tilde{\overline{\mathcal{P}}}^{\alpha\beta}\Tilde{z}^{\alpha}_{t+1},\sum_{\alpha\in\mathcal{A}}\Tilde{\overline{\mathcal{P}}}^{\alpha\beta}\Tilde{z}^{\alpha}_{t+1}]}{\sum\limits_{\alpha\in\mathcal{A}}\!\mathbb{E}_{\Tilde{\overline{\mathcal{P}}}^{\alpha\beta}}[\Tilde{\overline{\mathcal{P}}}^{\alpha\beta}]^2(\Tilde{z}^{\alpha}_{t+1})^2} \\& + \frac{\mathbb{E}_{\Tilde{\overline{\mathcal{P}}}^{\alpha\beta}}[\Tilde{\overline{\mathcal{P}}}^{\alpha\beta}]\Tilde{z}^{\alpha}_{t+1}\text{Var}[\sum_{\alpha\in\mathcal{A}}\Tilde{\overline{\mathcal{P}}}^{\alpha\beta}\Tilde{z}^{\alpha}_{t+1}]}{\sum\limits_{\alpha\in\mathcal{A}}\!\mathbb{E}_{\Tilde{\overline{\mathcal{P}}}^{\alpha\beta}}[\Tilde{\overline{\mathcal{P}}}^{\alpha\beta}]^3(\Tilde{z}^{\alpha}_{t+1})^3}  = \frac{\overline{\mathcal{P}}^{\alpha\beta}\Tilde{z}^{\alpha}_{t+1}}{\sum\limits_{\alpha\in\mathcal{A}}\!\overline{\mathcal{P}}^{\alpha\beta}\Tilde{z}^{\alpha}_{t+1}}  \\
& -\frac{(\Tilde{z}^{\alpha}_{t+1})^2\text{Cov}[\Tilde{\overline{\mathcal{P}}}^{\alpha\beta}\!,\Tilde{\overline{\mathcal{P}}}^{\alpha\beta}\!]\! +\!\!\!\! \sum\limits_{\nu\neq\alpha\in\mathcal{A}}\!\!\!\!\!\Tilde{z}^{\alpha}_{t+1}\Tilde{z}^{\nu}_{t+1}\text{Cov}[\Tilde{\overline{\mathcal{P}}}^{\alpha\beta}\!,\Tilde{\overline{\mathcal{P}}}^{\nu\beta}]}{\sum\limits_{\alpha\in\mathcal{A}}(\overline{\mathcal{P}}^{\alpha\beta}\Tilde{z}^{\alpha}_{t+1})^2} \\& + \frac{\splitfrac{\overline{\mathcal{P}}^{\alpha\beta}\Tilde{z}^{\alpha}_{t+1}\Big\{\sum\limits_{\alpha\in\mathcal{A}}(\Tilde{z}^{\alpha}_{t+1})^2\text{Var}[\Tilde{\overline{\mathcal{P}}}^{\alpha\beta}]+\sum\limits_{\alpha\in\mathcal{A}}\sum\limits_{\nu
\neq\alpha\in\mathcal{A}}}{\Tilde{z}^{\alpha}_{t+1}\Tilde{z}^{\nu}_{t+1}\text{Cov}[\Tilde{\overline{\mathcal{P}}}^{\alpha\beta},\Tilde{\overline{\mathcal{P}}}^{\nu\beta}]\Big\}}}{\sum\limits_{\alpha\in\mathcal{A}}(\overline{\mathcal{P}}^{\alpha\beta}\Tilde{z}^{\alpha}_{t+1})^3},
\end{split} \label{avg_optimal_policy_temp}
\end{align}
where:
\begin{align}
& \text{Cov}[\Tilde{\overline{\mathcal{P}}}^{\alpha\beta},\Tilde{\overline{\mathcal{P}}}^{\alpha\beta}] = \text{Var}[\Tilde{\overline{\mathcal{P}}}^{\alpha\beta}] = \frac{\overline{\mathcal{P}}^{\alpha\beta}(1-\overline{\mathcal{P}}^{\alpha\beta})}{k+1} \label{avg_cost_cov_same_1}
\end{align}
\begin{align}
\begin{split}
& \text{Cov}[\Tilde{\overline{\mathcal{P}}}^{\alpha\beta},\Tilde{\overline{\mathcal{P}}}^{\nu\beta}] = \frac{-(k\overline{\mathcal{P}}^{\alpha\beta})(k\overline{\mathcal{P}}^{\nu\beta})}{(\sum\limits_{\alpha\in\mathcal{A}}k\overline{\mathcal{P}}^{\alpha \beta})^{2}(\sum\limits_{\alpha\in\mathcal{A}}k\overline{\mathcal{P}}^{\alpha \beta}+1)} \\& = \frac{-k^2\overline{\mathcal{P}}^{\alpha\beta}\overline{\mathcal{P}}^{\nu\beta}}{(k\sum\limits_{\alpha\in\mathcal{A}}\overline{\mathcal{P}}^{\alpha \beta})^{2}(k\sum\limits_{\alpha\in\mathcal{A}}\overline{\mathcal{P}}^{\alpha \beta}+1)} = \frac{-\overline{\mathcal{P}}^{\alpha\beta}\overline{\mathcal{P}}^{\nu\beta}}{k+1}
\end{split} \label{avg_cost_cov_diff_1}
\end{align}
By using \eqref{avg_cost_cov_same_1} and \eqref{avg_cost_cov_diff_1} in \eqref{avg_optimal_policy_temp}, we conclude the proof.

\vspace{10mm}

\section{Proof of Corollary \ref{average_expected_cost}} \label{appendix_proof_average_cost}
We start with the non-private LS-MDP and $\forall \beta\in\mathcal{A},t\in\mathcal{T}$:
\begin{align}
\begin{split}
& \varphi^{\beta}_{t} = -U_{t}^{\beta} - \gamma \log \sum_{\alpha \in \mathcal{A}}\overline{\mathcal{P}}^{\alpha \beta}z_{t+1}^{\alpha}
\end{split} \label{Val_func_det_cost_avg}
\end{align}
Similarly, for the DP LS-MDP and $\forall \beta\in\mathcal{A},t\in\mathcal{T}$:
\begin{align}
\begin{split}
& \Tilde{\varphi}^{\beta}_{t} = -U_{t}^{\beta} + \!\gamma\!\!\sum_{\alpha \in \mathcal{A}}\!\Tilde{\mathcal{P}}_{t}^{\alpha\beta}(  \log \Tilde{\mathcal{P}}_{t}^{\alpha\beta} -  \log {\overline{\mathcal{P}}^{\alpha\beta}} +\frac{\Tilde{\varphi}^{\alpha}_{t+1}}{\gamma}) 
\end{split} \label{Val_func_dp_cost_avg}
\end{align}
Then, the expected cost of privacy is derived as:
\begin{align}
\begin{split}
& \Delta C = \mathbb{E}_{\Tilde{\overline{\mathcal{P}}}^{\alpha\beta}}[\Delta \varphi_{t}^{\beta}] = \mathbb{E}_{\Tilde{\overline{\mathcal{P}}}^{\alpha\beta}}[\Tilde{\varphi_{t}}^{\beta}] - \varphi_{t}^{\beta} = \\& \gamma\!\!\sum_{\alpha \in \mathcal{A}}\!\mathbb{E}_{\Tilde{\overline{\mathcal{P}}}^{\alpha\beta}}[\Tilde{\mathcal{P}}_{t}^{\alpha\beta}\log \Tilde{\mathcal{P}}_{t}^{\alpha\beta}] - \gamma\!\!\sum_{\alpha \in \mathcal{A}}\!\mathbb{E}_{\Tilde{\overline{\mathcal{P}}}^{\alpha\beta}}[\Tilde{\mathcal{P}}_{t}^{\alpha\beta}]\log {\overline{\mathcal{P}}^{\alpha\beta}} \\& - \gamma\!\!\sum_{\alpha \in \mathcal{A}}\!\mathbb{E}_{\Tilde{\overline{\mathcal{P}}}^{\alpha\beta}}[\Tilde{\mathcal{P}}_{t}^{\alpha\beta}]\Tilde{z}_{t+1}^{\alpha} + \gamma \log \sum_{\alpha \in \mathcal{A}}\overline{\mathcal{P}}^{\alpha \beta}z_{t+1}^{\alpha},
\end{split} \label{cost_dp_avg_temp}
\end{align}
where $\mathbb{E}_{\Tilde{\overline{\mathcal{P}}}^{\alpha\beta}}[ \Tilde{\mathcal{P}}_{t}^{{\alpha \beta}}]$ is already computed and given by \eqref{expected_policy_average}, and we derive $\mathbb{E}_{\Tilde{\overline{\mathcal{P}}}^{\alpha\beta}}[\Tilde{\mathcal{P}}_{t}^{\alpha\beta}\log \Tilde{\mathcal{P}}_{t}^{\alpha\beta}]$ as follows:
\begin{align}
\begin{split}
& \mathbb{E}_{\Tilde{\overline{\mathcal{P}}}^{\alpha\beta}}[\Tilde{\mathcal{P}}_{t}^{\alpha\beta}\log \Tilde{\mathcal{P}}_{t}^{\alpha\beta}] \approx \mathbb{E}_{\Tilde{\overline{\mathcal{P}}}^{\alpha\beta}}[(\Tilde{\mathcal{P}}_{t}^{\alpha\beta})^2-\Tilde{\mathcal{P}}_{t}^{\alpha\beta}] \\&
= \mathbb{E}_{\Tilde{\overline{\mathcal{P}}}^{\alpha\beta}}[(\Tilde{\mathcal{P}}_{t}^{\alpha\beta})^2]-\mathbb{E}_{\Tilde{\overline{\mathcal{P}}}^{\alpha\beta}}[\Tilde{\mathcal{P}}_{t}^{\alpha\beta}] \\& = \text{Var}(\Tilde{\mathcal{P}}_{t}^{\alpha\beta})
+\mathbb{E}_{\Tilde{\overline{\mathcal{P}}}^{\alpha\beta}}[\Tilde{\mathcal{P}}_{t}^{\alpha\beta}]^2-\mathbb{E}_{\Tilde{\overline{\mathcal{P}}}^{\alpha\beta}}[\Tilde{\mathcal{P}}_{t}^{\alpha\beta}]
\end{split} \label{cost_dp_avg_temp2}
\end{align}
Now, we evaluate $\text{Var}(\Tilde{\mathcal{P}}_{t}^{\alpha\beta})$ in \eqref{cost_dp_avg_temp2} as follows:
\begin{align}
\begin{split}
& \text{Var}(\Tilde{\mathcal{P}}_{t}^{\alpha\beta}) = \text{Var}\Bigg(\!\frac{\Tilde{\overline{\mathcal{P}}}^{\alpha\beta}\Tilde{z}^{\alpha}_{t+1}}{\sum\limits_{\alpha\in\mathcal{A}}\textcolor{black}{\Tilde{\overline{\mathcal{P}}}^{\alpha\beta}}\Tilde{z}^{\alpha}_{t+1}} \!\Bigg)  \\& \approx
\frac{\text{Var}(\Tilde{\overline{\mathcal{P}}}^{\alpha\beta}\Tilde{z}^{\alpha}_{t+1})}{\sum\limits_{\alpha\in\mathcal{A}}\!\mathbb{E}_{\Tilde{\overline{\mathcal{P}}}^{\alpha\beta}}[\Tilde{\overline{\mathcal{P}}}^{\alpha\beta}]^2(\Tilde{z}^{\alpha}_{t+1})^2} \\& - \frac{2\mathbb{E}_{\Tilde{\overline{\mathcal{P}}}^{\alpha\beta}}[\Tilde{\overline{\mathcal{P}}}^{\alpha\beta}]\Tilde{z}^{\alpha}_{t+1}\text{Cov}[\Tilde{\overline{\mathcal{P}}}^{\alpha\beta}\Tilde{z}^{\alpha}_{t+1},\sum_{\alpha\in\mathcal{A}}\Tilde{\overline{\mathcal{P}}}^{\alpha\beta}\Tilde{z}^{\alpha}_{t+1}]}{\sum\limits_{\alpha\in\mathcal{A}}\!\mathbb{E}_{\Tilde{\overline{\mathcal{P}}}^{\alpha\beta}}[\Tilde{\overline{\mathcal{P}}}^{\alpha\beta}]^3(\Tilde{z}^{\alpha}_{t+1})^3} \\& + \frac{\mathbb{E}_{\Tilde{\overline{\mathcal{P}}}^{\alpha\beta}}[\Tilde{\overline{\mathcal{P}}}^{\alpha\beta}]^2(\Tilde{z}^{\alpha}_{t+1})^2\text{Var}[\sum_{\alpha\in\mathcal{A}}\Tilde{\overline{\mathcal{P}}}^{\alpha\beta}\Tilde{z}^{\alpha}_{t+1}]}{\sum\limits_{\alpha\in\mathcal{A}}\!\mathbb{E}_{\Tilde{\overline{\mathcal{P}}}^{\alpha\beta}}[\Tilde{\overline{\mathcal{P}}}^{\alpha\beta}]^4(\Tilde{z}^{\alpha}_{t+1})^4}
\end{split} \label{cost_dp_avg_temp3}
\end{align}
We have already computed terms  $\mathbb{E}_{\Tilde{\overline{\mathcal{P}}}^{\alpha\beta}}\big[\cdot \big],   \text{Var}\big[\cdot \big],$ and $\text{Cov}\big[\cdot \big]$ used in \eqref{cost_dp_avg_temp3} as part of Appendix \ref{appendix_proof_expected_policy}. By using respective terms in \eqref{cost_dp_avg_temp3}, inserting \eqref{cost_dp_avg_temp3} in \eqref{cost_dp_avg_temp2}, and then using \eqref{cost_dp_avg_temp2} in \eqref{cost_dp_avg_temp}, we conclude the proof.

\end{document}